\documentclass[11pt, letterpaper, onecolumn]{IEEEtran}

\usepackage[ colorlinks = true,
             linkcolor = blue,
             urlcolor  = blue,
             citecolor = red,
             anchorcolor = magenta,
]{hyperref}

\usepackage{amssymb,amsmath,amsthm}
\usepackage{cite}

\theoremstyle{plain}  
\newtheorem{lemma}{Lemma}

\newtheorem{proposition}[lemma]{Proposition}
\newtheorem{corollary}[lemma]{Corollary}
\newtheorem{theorem}[lemma]{Theorem}
\newtheorem{remark}{Remark}
\theoremstyle{definition}
\newtheorem{example}{Example}

\newcommand{\cX}{\mathcal{X}}
\newcommand{\cY}{\mathcal{Y}}
\newcommand{\cZ}{\mathcal{Z}}
\newcommand{\cM}{\mathcal{M}}
\newcommand{\cA}{\mathcal{A}}
\newcommand{\cP}{\mathcal{P}}
\newcommand{\cS}{\mathcal{S}}
\newcommand{\cC}{\mathcal{C}}

\newcommand{\cT}{\mathcal{T}}
\newcommand{\cB}{\mathcal{B}}

\newcommand{\abs}[1]{|#1|}
\newcommand{\eps}{\varepsilon}

\newcommand{\proj}[1]{|k\rangle\!\langle k|}

\newcommand{\twonorm}[1]{\| #1 \|_{2}}

\newcommand{\cv}{\textrm{cv}}
\newcommand{\Ceps}{C_{\eps}}

\DeclareMathOperator*{\plimsup}{\mathfrak{p}-lim\, sup\,}
\DeclareMathOperator*{\pliminf}{\mathfrak{p}-lim\, inf\,}

\DeclareMathOperator{\Exp}{E}
\DeclareMathOperator*{\Var}{Var}
\DeclareMathOperator*{\Cov}{Cov}

\DeclareMathOperator*{\argmax}{arg\,max}
\DeclareMathOperator*{\argmin}{arg\,min}
\DeclareMathOperator{\diag}{diag}

\begin{document}
\flushbottom

\title{Second-Order Coding Rates for Channels with State} 
 
 \author{Marco Tomamichel$^*$, \emph{Member, IEEE},   \hspace{.025in}  and \hspace{.025in}   Vincent Y.~F.\ Tan$^\dagger$, \emph{Member, IEEE}\thanks{$^*$ Center for Quantum Technologies, National University of Singapore (Email: {cqtmarco@nus.edu.sg}).}  \thanks{$^\dagger$   Department of Electrical and Computer Engineering (ECE) and Department of Mathematics, National University of Singapore (NUS), (Email:
  {vtan@nus.edu.sg}). }\thanks{This paper was presented in part at the IEEE Information Theory Workshop (ITW 2013), Seville, Spain \cite{TomTanITW2012}.}}
   

\maketitle

\begin{abstract}
We study the performance limits of state-dependent discrete memoryless channels with a discrete state available at both  the encoder and  the decoder.
We establish the $\eps$-capacity as well as necessary and sufficient conditions for the strong converse property for such channels when the sequence of channel states is not necessarily stationary, memoryless or ergodic. We then seek a finer characterization of these capacities in terms of second-order coding rates. The general results are supplemented by several examples including i.i.d.\  and Markov states and mixed channels. 
\end{abstract}

\section{Introduction}
We revisit  the   classical problem of channel coding with random states~\cite[Ch.~7]{elgamal} where the channel, viewed as a stochastic kernel from the set of inputs $\cX$ and states $\cS$ to the output $\cY$, is discrete, memoryless and stationary, while the {\em discrete} state is allowed to be a {\em general   source} in the sense of Verd\'u-Han~\cite{VH94, Han10}. This means that, apart from the state having a finite alphabet, it does not have to be stationary, memoryless nor ergodic. The state is  known noncausally at both encoder and decoder. See Fig.~\ref{fig:states}. This models the scenario in which the channel viewed as a stochastic kernel from $\cX$ to $\cY$ is possibly non-ergodic or having memory, where both non-ergodicity and memory are induced by the  general state sequence. We derive the $\eps$-capacity and its optimistic version~\cite{chen99} as well as second-order coding rates~\cite{Hayashi08,Hayashi09} and specialize the general results to various state distributions, such as  independent and identically distributed (i.i.d.) and Markov.  The justifications of  the second-order results  require new techniques such as multiple applications of various forms of the Berry-Esseen theorem~\cite[Sec.~2.2]{tao12} \cite[Sec.\ XVI.7]{feller} to simplify expectations of Gaussian cumulative distributions functions.

\subsection{Main Results and Technical Contributions}
There are two sets of results in this paper, namely the results concerning first- and second-order coding rates.

First, we derive general formulas for the $\eps$-capacity and optimistic $\eps$-capacity (in the sense of Chen-Alajaji~\cite{chen99}) both under cost constraints. We do so by using the information spectrum method by Han and Verd\'u~\cite{VH94, Han10}. The direct part is proved using an extension of Feinstein's lemma~\cite{feinstein, Lan06}  to channels with state while the converse part is proved using a   one-shot converse bound established by the present authors~\cite{tomamicheltan12}.  These capacities only depend on the cumulative distribution function (cdf) of the C\`esaro mean of the  capacity-cost functions which can also be expressed as an average of the  capacity-cost functions with respect to the empirical distribution of the state sequence. This corroborates our intuition because the channel is well-behaved, thus it does not require characterization using information spectrum quantities and  probabilistic limits~\cite{VH94, Han10}. The only complication that can arise is due to the generality of the state and for this, we do require probabilistic limits.  We thus observe a neat decoupling of the randomness induced by the channel and the state.  By specializing the $\eps$-capacity and optimistic $\eps$-capacity to the capacity and optimistic capacity respectively, we derive a necessary and sufficient condition for the strong converse  \cite[Sec.~V]{VH94} \cite[Def.~3.7.1]{Han10} to hold.  By a further application of Chebyshev's inequality, we provide a simpler sufficient condition for the strong converse to hold. This condition is  based  only on first- and second-order statistics of the state process and hence, is much easier to verify. We provide examples to illustrate the various conditions for the strong converse property to hold. 

Second, we use the one-shot bounds to derive optimum second-order coding rates~\cite{Hayashi08,Hayashi09} for this problem of channels with general state available at both encoder and decoder. Second-order coding rates  provide an approximate characterization of the backoff from capacity at blocklength $n$. These rates are typically derived via one application of the central limit theorem (for each of the direct and converse parts) and hence are usually   expressed in terms of a variance or dispersion quantity~\cite{PPV10,PPV11} and the inverse of the cdf of a standard Gaussian. For example, in channel coding, Strassen showed that the maximum size of a codebook $M^*(W^n,\eps)$ for which there exist  codes  of length $n$ and average    error probability  $\eps$ for a well-behaved discrete memoryless channel (DMC) $W$ satisfies
\begin{equation}
\log M^*(W^n,\eps) = nC(W) + \sqrt{nV_\eps(W)} \Phi^{-1}(\eps) + O(\log n), \label{eqn:strass}
\end{equation}
where $C(W)$ and $V_\eps(W)$ are the channel capacity and $\eps$-channel dispersion \cite{PPV10} of $W$ respectively. The {\em second-order coding rate}, a term coined by Hayashi~\cite{Hayashi09}, is coefficient of the $\sqrt{n}$  in \eqref{eqn:strass}, namely $\sqrt{V_\eps(W)} \Phi^{-1}(\eps)$.   We would like to characterize this quantity for channels with state.   We first allow the state to be general  but finitely-valued and derive a general formula (in the Verd\'u-Han sense~\cite{VH94}) for the optimum second-order coding rate. Subsequently, this result is specialized to various state distributions including i.i.d.\ states, Markov states and the mixed channels scenario, previously studied by Polyanskiy-Poor-Verd\'u~\cite[Thm.~7]{PPV11}. 

To illustrate our  main contribution at a high level, let us consider the states being i.i.d.\ with distribution $P_S\in\cP(\cS)$. The capacity in this case is $C(W,P_S) = \max_{P\in\cP(\cX|\cS)}I(P;W|P_S)$ \cite[Sec.~7.4.1]{elgamal}.  Suppose first that the state is known to be some deterministic  sequence $s^n$ of type  \cite{Csi97} (empirical distribution)  $T_{s^n}$.  Denote the optimum error probability for a length-$n$ block code with $M$ codewords as $\eps^*(W^n,M,s^n)$.  By a slight extension of  Strassen's channel coding result in \eqref{eqn:strass}  to  memoryless but non-stationary channels we find that, for typical $s^n$, 
\begin{equation}
\eps^*(W^n,M,s^n)= \Phi\bigg( \frac{\log M -n C(T_{s^n}) }{ \sqrt{nV(T_{s^n})}}\bigg) + O\bigg( \frac{1}{\sqrt{n} }\bigg), \nonumber
\end{equation}
where the {\em empirical capacity} and {\em  empirical dispersion} are respectively defined as
\begin{equation}
 C(T_{s^n}) :=\frac{1}{n}\sum_{i=1}^n C(W_{s_i}),\quad\mbox{and}\quad  V(T_{s^n}) :=\frac{1}{n}\sum_{i=1}^n V(W_{s_i}). \nonumber
\end{equation}
We have assumed that the $\eps$-dispersion of each channel, $W_s$, $s\in\cS$, is positive and does not depend on $\eps$, an assumption that is true for almost all DMCs. Denote this  dispersion as $V(W_s)$.
The optimum error probability when the state is random and i.i.d.\ is  denoted as $\eps^*(W^n,M)$ and  it can be written as the following expectation:
\begin{equation}
\eps^*(W^n,M ) = \Exp_{S^n}\big[ \eps^*(W^n,M,S^n)\big]=\Exp_{S^n} \bigg[  \Phi\bigg( \frac{\log M -n   C(T_{S^n}) }{ \sqrt{n   V(T_{s^n})  }}\bigg)  \bigg].\label{eqn:epsM}
\end{equation}
So the analysis of the expectation above is crucial.  One of our main technical contributions is to show that    this expectation roughly equals 
\begin{equation}
    \Phi\bigg( \frac{\log M -n C(W,P_S)  }{ \sqrt{n  ( \Var[C(W_S)] + \Exp[V(W_S)]) }}\bigg) .\label{eqn:simplify}
\end{equation}
This shows that the dispersion of the channel with state is $\Var[C(W_S)] + \Exp[V(W_S)]$. The justification that  \eqref{eqn:epsM} is approximately \eqref{eqn:simplify} is done in Lemmas~\ref{lem:approx_V} and~\ref{lem:approx_sum_Vs} where we first approximate $V(T_{S^n})$ with $\Exp[V(W_S)]$ (without too much loss) and subsequently approximate $C(T_{S^n})$  with the true capacity, $C(W,P_S)$. The second approximation results in the additional variance term.  This variance term, $\Var[C(W_S)]$, represents the randomness of the state while the other variance term, $\Exp[V(W_S)]$, represents the randomness of the channel given the state. 

\subsection{Related Work}
Channels with random states   have been studied extensively. See the book by El Gamal and Kim~\cite[Ch.~7]{elgamal}  for a thorough overview. We use the information spectrum method to analyze the problem of channels with random states where the state can be general.  By placing the distribution of the information density random variable in a central role, Han and Verd\'u~\cite{VH94, Han10}  treated information-theoretic  problems beyond the i.i.d.\ or information-stable setting. To the best of our knowledge, the only other work that analyzes channels with  state from the information spectrum viewpoint is that in~\cite{Tan13b} for the Gel'fand-Pinsker problem (i.e.\ the state is only available at the encoder).

Recently, there has been a surge of interest in second-order coding rates for a variety of information-theoretic tasks such as  source coding~\cite{Kot97, TK14}, intrinsic randomness~\cite{Hayashi08} and channel coding~\cite{PPV10,Hayashi09}. This line of work, in fact, started from Strassen's seminal work \cite{strassen} on hypothesis testing and channel coding in which he characterized the fundamental limits up to the second-order. There are two other noteworthy works that are closer in spirit to our study of the second-order coding rates for channels with state. First, we mention the work of Polyanskiy-Poor-Verd\'u~\cite{PPV11} who derived the dispersion of the Gilbert-Elliott channel~\cite{Gilbert, Elliott, Mush} where the state is either unavailable to both  terminals or only available at the decoder.   Second, the work by Ingber and Feder~\cite{ingber10} involves finding the dispersion for the problem where the  i.i.d.\ state is available only at the decoder. We   compare and contrast our results to the relevant results in~\cite{PPV11} and~\cite{ingber10}. Finally, we mention that there are some  recent works on second-order coding rates involving  SIMO \cite{yang13} and MIMO  \cite{hoydis13} fading channels but  the setup in this paper is different from that in~\cite{yang13,hoydis13}. In particular, we consider a discrete state setup here and this leads to different results compared to~\cite{yang13} where it was observed that the dispersion term is zero.

\subsection{Paper Organization}
This paper is structured as follows: In Section~\ref{sec:prelims}, we state the definitions of the information-theoretic problem, capacities,  optimistic capacities and second-order coding rates.  The first-order and strong converse  results are presented in Section~\ref{sec:cap}. The second-order results are presented next in Section~\ref{sec:second}. These sections will be supplemented with five continuing examples that illustrate specializations of the general formulas. The proofs are deferred to the latter sections. In Section~\ref{sec:one-shot}, we introduce and prove two one-shot results that are used to prove subsequent direct and converse parts. In Section~\ref{sec:prf_eps_cap}, we prove the main first-order result in Theorem~\ref{thm:eps_cap}.  In Section~\ref{sec:prf_second}, we use the two one-shot results to prove the general result concerning the second-order coding rate of Theorem~\ref{thm:sec_gen}. Finally, Section~\ref{sec:prf_examples} contains the proofs of the second-order coding rates for the various examples in Theorems~\ref{thm:mixed}--\ref{thm:weird}.

\section{Preliminaries and Definitions} \label{sec:prelims}
\subsection{Basic Definitions}
We assume throughout that $\cX$, $\cY$ and $\cS$ are finite sets. Let $\cP(\cX)$ be the set of probability distributions on $\cX$. We also denote the set of channels from $\cX$ to $\cY$ as $\cP(\cY|\cX) \cong \cP(\cY)^{|\cX|}$. In the following, we let $W\in\cP(\cY|\cX\times\cS)$ be a channel where $\cX$ denotes the input alphabet, $\cS$ denotes the state alphabet and $\cY$ denotes the output alphabet. The set of all $x\in\cX$ that are {\em admissible} for the channel in state $s\in\cS$ is 
\begin{equation*}
\cB_{s}(\Gamma) := \{ x \in \cX \,|\, b_s(x)\le\Gamma \}
\end{equation*}
for some functions $b_s:\cX\to\mathbb{R}^+$ and $\Gamma>0$. We do not explicitly mention $\Gamma$ if there are no cost constraints, i.e., if $\Gamma = \infty$. The channel state $S$ is a random variable with probability distribution $P_S\in\cP(\cS)$.

For any $P\in\cP(\cX|\cS)$, we define the conditional distribution $PW\in\cP(\cY|\cS)$ as $PW(y|s) :=\sum_x P(x|s) W(y|x,s)$. The following {\em conditional log-likelihood ratios} are of interest:
\begin{align*}
i(x;y|s) :=\log\frac{W(y|x,s)}{PW(y|s)} ,  \quad 
j_Q(x;y|s) :=\log\frac{W(y|x,s)}{Q(y|s)} , 
\end{align*}
where the latter definition applies for any $Q\in\cP(\cY|\cS)$ with $Q(\cdot|s)\gg W(\cdot|x,s)$ for every $(x,s)\in\cX\times\cS$.\footnote{The notation $Q \gg P$ denotes the fact that $P$ is absolutely continuous with respect to $Q$. We also say that $Q$ dominates $P$.}
  We denote the  {\em conditional mutual information} as $I(P,W|P_S):=\Exp[i(X;Y|S)]$ where $(S, X, Y)$ is distributed using the law $(S, X, Y) \leftarrow P_S(s)P(x|s) W(y|x,s)$.  Furthermore, the \emph{capacity-cost function} of the channel $W_s := W(\cdot|*,s)$\footnote{The notation $W_s = W(\cdot|*,s)$ is a shorthand for the statement that $W_s(y|x) = W(y|x,s)$ for all $(s,x,y) \in\cS \times  \cX \times \cY$.} is defined as (in bits per channel use)
\begin{align*}
  C_s(\Gamma) := \max_{P\in \cP(\cB_s(\Gamma))}I(P,W_s), \qquad \textrm{where} \quad I(P,W) :=  \sum_{x \in \cX} P(x) D(W(\cdot|x)\|PW)
\end{align*}  
and $D(P\|Q)$ denotes the relative entropy between $P$ and $Q$.
The \emph{average capacity-cost function} given a probability distribution $P_S \in \cP(\cS)$ is defined as 
\begin{align*}
C(P_S, \Gamma):=\sum_{s\in\cS} P_S(s) C_s(\Gamma).
\end{align*}
The quantities $C_s(\Gamma)$ and $C(P_S,\Gamma)$ will be respectively denoted as  $C_s$ and $C (P_S )$ if there are no cost constraints.

For our results concerning second-order coding rates, we do not consider cost constraints (i.e., $\Gamma = \infty$). Then,
we denote the maximizing distribution in $C_s = \max_{P \in \cP(\cX) } I(P, W_s)$ as $P_s^* \in \cP(\cX)$ and we assume that $P_s^*$ is unique for each $s\in\cS$. For a precise justification of this assumption, see Section~\ref{sec:second}. Furthermore, we define the   conditional distribution  $P^*\in\cP(\cX|\cS)$ as $P^*(x|s)=P_s^*(x)$. We   also consider the following second-order quantities. Define the {\em information dispersion} for the channel $W_s \in \cP(\cY|\cX)$ as 
\begin{align*}
V_s  := V(P_s^*, W_s), \quad \textrm{where} \quad
V(P, W) :=  \sum_{x\in\cX}P(x) \sum_{y\in\cY} W(y|x) \left[ \log\frac{ W(y|x)}{PW(y)}- D(W( \cdot|x) \|P W ) \right]^2,
\end{align*}
is the {\em conditional information variance}.
The quantity $V_s$ is also known more simply as the {\em dispersion}  of the channel $W_s$  in the literature~\cite{PPV10,PPV11} although we prefer to reserve the term ``dispersion'' to be an operational quantity. See Section~\ref{sec:defs_sec}. The {\em average conditional information variance} with respect to $P_S$ is   
\begin{align*}
V(P_S) := \sum_{s\in\cS} P_S(s) V_s  . 
\end{align*}

For  a sequence of real-valued  random variables $ \{A_n\}_{n=1}^{\infty}$, the  {\em $\liminf$ and $\limsup$ in probability}~\cite{Han10,VH94,chen99} are respectively defined as 
\begin{align*}
\pliminf_{n\to\infty}A_n&:=\sup\big\{ a \in \mathbb{R} \, \big|\,  \limsup_{n\to\infty}\Pr[  A_n<a]=0\big\} ,\quad\mbox{and} \\
\plimsup_{n\to\infty}A_n&:= -\pliminf_{n\to\infty}(-A_n) = \inf\big\{ a \in \mathbb{R} \,\big|\, \liminf_{n\to\infty} \Pr [ A_n \leq a ] = 1 \big\}.
\end{align*}

Let the probability density function (pdf) of the standard normal distribution be denoted as  
$$
\phi(x):= \frac{1}{\sqrt{2 \pi}} \exp \Big(-\frac{1}{2} x^2\Big).
$$ 
We also extensively  employ the cumulative distribution function (cdf) of the standard normal distribution 
$$\Phi(a) := \int_{-\infty}^a \phi(x) \,\mathrm{d}x.$$ 
We define its inverse as $\Phi^{-1}(\eps) := \sup \{ a \in \mathbb{R} \,|\, \Phi(a) \leq \eps \}$, which evaluates to the usual inverse for $0 < \eps < 1$ and continuously extended to take values $\pm \infty$ outside that range.

\subsection{Codes for Channels with States}
A {\em code} for the channel $W$ with cost constraint $\Gamma$ 
is defined by  $\cC:=\{\cM,e,d\}$ where $\cM$ is the  {\em message set}, $e :\cM\times\cS\to\cX$ is the {\em  encoder} and   $d : \cY\times\cS\to\cM$ is the {\em  decoder}. The encoder must satisfy $e(m,s) \in \cB_s(\Gamma)$ for all $s \in \cS$ and $m \in \cM$. For $S \leftarrow P_S$, the average (for uniform $M$) and maximum error probabilities are respectively defined as 
\begin{align*}
p_{\mathrm{avg}}(\cC;W,P_S) &:=\Pr[M\ne M'] \quad \textrm{and}\\
p_{\mathrm{max}}(\cC;W,P_S) &:=\max_{m\in\cM}\Pr[M\ne M' |M=m] ,
\end{align*}
The relation of the random variables $M$, $X$, $Y$, $M'$ and $S$ is depicted in Fig.~\ref{fig:states}.

We let $M^*(\eps,\Gamma;W, P_S)$ 
be the maximum \emph{code size} $|\cM|$ for which transmission with average 
error probability of at most~$\eps$ is possible through the channel $W$ when the state with distribution $P_S$ is known at both encoder and decoder, i.e.
\begin{align*}
  M^*(\eps, \Gamma; W, P_S) = \sup \big\{ k \in \mathbb{N} \,\big|\, \exists\, \cC = \{\cM, e, d\} \textrm{ with } |\cM| = k \textrm{ and } p_{\textrm{avg}}(\cC; W, P_S) \leq \eps \big\} .
\end{align*}

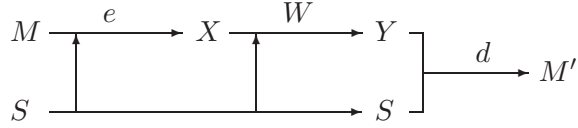
\begin{figure}
\vspace{0.2cm} 
\centering
\begin{picture}(205,48)

\put(5,40){\mbox{$M$}}
\put(5,10){\mbox{$S$}}

\put(20,43){\vector(1,0){50}}
\put(40,47){\mbox{$e$}}
\put(20,13){\vector(1,0){118}}

\put(30,13){\vector(0,1){30}}

\put(75,40){\mbox{$X$}}

\put(88,43){\vector(1,0){50}}
\put(108,47){\mbox{$W$}}
\put(98,13){\vector(0,1){30}}

\put(143,40){\mbox{$Y$}}
\put(143,10){\mbox{$S$}}

\put(156,13){\line(1,0){5}}
\put(156,43){\line(1,0){5}}
\put(161,13){\line(0,1){30}}
\put(161,28){\vector(1,0){40}}
\put(181,32){\mbox{$d$}}

\put(205,25){\mbox{$M'$}}

\end{picture}

\vspace{-0.4cm}
\caption{Casual dependence of random variables when state $S$ is known at encoder and decoder, $M$ is the message, $X$ and $Y$ are the cannel input and outputs, respectively, and $M'$ is the estimate of the message.}
\label{fig:states}
\end{figure}

For a length-$n$ sequence $s^n = (s_1, \ldots, s_n) \in \cS^n$, the {\em type}~\cite[Ch.~2]{Csi97} or {\em empirical distribution} of $s^n$, denoted $T_{s^n}\in \cP(\cS)$, is the relative frequency of various symbols  in  $s^n$. More precisely, $T_{s^n}(s) =\frac{1}{n} \sum_{k=1}^n 1\{   s_k = s\}$. The set of all types with alphabet $\cS$ formed from sequences of length $n$ is denoted as $\cP_n(\cS)$. Note that the type-counting lemma~\cite{Csi97} states that $|\cP_n(\cS)|\le (n+1)^{|\cS|}$.

When we consider $n$ uses of the channel, $W^n \in \cP(\cY^n|\cX^n\times\cS^n)$ observes the following memoryless behavior
\begin{equation*}
W^n(y^n|x^n,s^n)=\prod_{k=1}^n W(y_k|x_k,s_k),\quad(x^n,y^n,s^n) \in \cX^n\times\cY^n\times \cS^n.
\end{equation*}
Moreover, admissible inputs $x^n$ when the channel  $W^n$ is in state $s^n$  must belong to the   set
\begin{align*}
\cB_{s^n}(\Gamma):=\bigg\{ x^n\in\cX^n \, \bigg|\, \frac{1}{n} \sum_{k=1}^n b_{s_k}(x_k) \leq \Gamma \bigg\}.
\end{align*}
To model general behavior of the channel, we allow the state sequence or source $\hat{S} :=\{ S^n=(S^{(n)}_1,\ldots, S^{(n)}_n)\}_{n=1}^{\infty}$ to evolve in an arbitrary manner in the sense of Verd\'u-Han~\cite{VH94, Han10}.

\subsection{Definition of  the $(\eps,\Gamma)$-Capacity and the Optimistic $(\eps,\Gamma)$-Capacity}

We say that a number $R \in \mathbb{R} \cup \{\infty\}$ is an {\em $(\eps,\Gamma)$-achievable rate} if there exists a sequence of non-negative numbers $\{\eps_n\}_{n=1}^{\infty}$ such that 
\begin{align*}
\liminf_{n\to\infty}\frac{1}{n}\log M^*(\eps_n,\Gamma;W^n,P_{\hat{S}})\ge R,\quad\mbox{and}\quad \limsup_{n\to\infty}\eps_n\le\eps.
\end{align*}
The {\em $\eps$-capacity-cost function} $C(\eps, \Gamma; W, P_{\hat{S}})$ for $\eps\in [0,1]$ is the supremum of  all $(\eps,\Gamma)$-achievable rates. The {\em capacity-cost function} is $C( \Gamma; W, P_{\hat{S}}):=C(0, \Gamma; W, P_{\hat{S}})$. Note that $C(1, \Gamma; W, P_{\hat{S}}) = \infty$ by definition. 

Similarly, a number  $R\in \mathbb{R} \cup \{-\infty\}$ is an {\em optimistic $(\eps,\Gamma)$-achievable rate} \cite{chen99} if there exists a sequence  of non-negative numbers $\{\eps_n\}_{n=1}^{\infty}$ for which\footnote{One might think that replacing the first $\liminf$ in~\eqref{eq:optimistic_enough} with a $\limsup$ will lead to a more optimistic achievable rate. However, by considering a sequence $\{\eps_n\}$ that alternates between $\eps_0 < \eps$ and $\eps_1 = 1$, one sees that \emph{any} $R > 0$ is achievable for such a modified definition.} 
\begin{align}
\liminf_{n\to\infty}\frac{1}{n}\log M^*(\eps_n,\Gamma;W^n,P_{\hat{S}})\ge R, \quad\mbox{and}\quad  \liminf_{n\to\infty}\eps_n < \eps. \label{eq:optimistic_enough}
\end{align}
The {\em optimistic $\eps$-capacity-cost function}  $C^\dagger(\eps, \Gamma; W, P_{\hat{S}})$ for $\eps\in [0,1]$  is the supremum of all   optimistic $(\eps,\Gamma)$-achievable rates.\footnote{In fact, Chen-Alajaji~\cite[Defs.~4.9/4.10]{chen99} define the $\eps$-optimistic capacity slightly differently resulting in their version of the  $\eps$-optimistic capacity being characterizable  at all but at most countably many $\eps\in [0,1]$. Our definition allows us to characterize the $\eps$-optimistic capacity  for all $\eps\in [0,1]$. See  a discussion of this issue in \cite[Sec.~IV]{VH94}. }  
 Note that $C^\dagger(\eps, \Gamma; W, P_{\hat{S}})$ can also be expressed as the infimum of all numbers $R  \in\mathbb{R}$ such that for every sequence $\{\eps_n\}_{n=1}^{\infty}$ the following implication holds:
\begin{align*}
\liminf_{n\to\infty}\frac{1}{n}\log M^*(\eps_n,\Gamma;W^n,P_{\hat{S}})\ge R \implies \liminf_{n\to\infty}\eps_n \ge \eps.
\end{align*}
We define the {\em optimistic capacity-cost function} as $C^\dagger(\Gamma; W,P_{\hat{S}}):=C^\dagger(1, \Gamma; W, P_{\hat{S}})$. Moreover, $C^{\dagger}(0, \Gamma; W, P_{\hat{S}}) = -\infty$ by definition.

Note that by the additivity of the cost function in the multi-letter setting,  $\Gamma\mapsto C(\eps, \Gamma; W, P_{\hat{S}})$ and $\Gamma\mapsto C^\dagger(\eps, \Gamma; W, P_{\hat{S}})$ are both concave and hence continuous for $\Gamma>0$.

A channel $W$ with   general state $\hat{S}$ {\em has the strong converse property}~\cite[Sec.~V]{VH94} \cite[Def.~3.7.1]{Han10}  if
\begin{align*}
C(\Gamma; W,P_{\hat{S}})=C^\dagger(\Gamma; W,P_{\hat{S}}),
\end{align*}
for all $\Gamma>0$. This is the form of the strong converse property stated in Hayashi-Nagaoka~\cite{hayashi03}. In other words, the strong converse property holds if and only if for every sequence $\{\eps_n\}_{n=1}^{\infty}$ for which
\begin{align*}
\liminf_{n\to\infty}  \frac{1}{n}\log M^*(\eps_n,\Gamma;W^n,P_{\hat{S}})> C(\Gamma; W,P_{\hat{S}}),   
\end{align*}
we have $\lim_{n\to\infty}\eps_n = 1$.

We note that even though the quantities above are defined based on the average error probability, all the results in the following also hold for the maximum error probability.

\subsection{Definition of the $(\eps,\beta)$-Optimum Second-Order Coding Rate and the $\eps$-Dispersion} \label{sec:defs_sec}

For the second-order results to be presented in Section~\ref{sec:second}, we do not consider cost constraints so $\Gamma=\infty$ and thus, this parameter is omitted from the following definitions. 

Let $(\eps ,\beta )\in (0,1) \times [\frac{1}{2},1)$. We say that an extended real number $r \in \mathbb{R} \cup \{\pm\infty\}$ is an  {\em $(\eps,\beta)$-achievable second-order coding rate} if there exists a sequence of non-negative numbers $\{\eps_n\}_{n=1}^{\infty}$ such that 
\begin{align*}
\liminf_{n\to\infty}\frac{1}{n^\beta} \left[\log M^*(\eps_n;W^n,P_{\hat{S}}) - n C(\eps ; W, P_{\hat{S}}) \right] \ge r,\quad\mbox{and}\quad \limsup_{n\to\infty}\eps_n\le\eps,
\end{align*}
where $C(\eps;W, P_{\hat{S}})$ is the $\eps$-capacity defined in the preceding section.  Note that if $r$ is an {\em $(\eps,\beta)$-achievable second-order coding rate} then there exists a sequence of length-$n$ codes with number of codewords  and error probability being $M_n$ and $\eps_n$ respectively and satisfying 
\begin{align*}
\log M_n = n C(\eps ; W, P_{\hat{S}})+ n^\beta  \, r  + o\big( n^\beta \big)\quad\mbox{and}\quad \limsup_{n\to\infty}\eps_n\le\eps.
\end{align*}
 The {\em $(\eps,\beta)$-optimum second-order coding rate}~\cite{Hayashi09} $\Lambda( \eps,\beta; W, P_{\hat{S}})$ is the supremum of all   $(\eps,\beta)$-achievable second-order coding rates. Note that, although termed a ``rate'', $\Lambda( \eps,\beta; W, P_{\hat{S}})$ can in fact be negative or infinite.  This convention was also used in Hayashi's works~\cite{Hayashi08,Hayashi09}.

We remark   that there exists at most one {\em critical exponent} $\beta^* \in [\frac{1}{2},1)$ such that $r^* := \Lambda(\eps,\beta^*; W, P_{\hat{S}}) \in \mathbb{R} \setminus \{ 0 \}$ is finite and non-zero.  Then, if $\beta > \beta^*$, we necessarily have $\liminf_{n\to\infty}\frac{1}{n^\beta} \left[\log M^*(\eps_n;W^n,P_{\hat{S}}) - n C(\eps ; W, P_{\hat{S}}) \right] = 0$; on the other hand, if $\beta < \beta^*$, the $\liminf$ is infinite. The $(\eps,\beta)$-optimum second-order coding rate is thus a more general definition than the $\eps$-dispersion (defined below). Indeed, we will  show  in  Theorem~\ref{thm:subblock} that  there exists specific state distributions $P_{\hat{S}}$ that result in a  finite $\Lambda( \eps,\beta; W, P_{\hat{S}})$     for  $\beta$    being an arbitrary number in $[\frac{1}{2},1)$ and not only $\beta=\frac{1}{2}$ as in previous work on dispersion and second-order coding rates (e.g., \cite{Hayashi08, Hayashi09}).

Finally, when $\beta=\frac{1}{2}$ above, we can define, for $\eps \in (0,1)\setminus \{\frac{1}{2} \}$,   the  {\em $\eps$-dispersion}~\cite[Def.~2]{PPV10} of the channel $W$ with general state distribution $P_{\hat{S}}$  as
\begin{align*}
\Upsilon(\eps; W , P_{\hat{S}}):=\limsup_{n\to\infty} \left( \frac{\log M^*(\eps;W^n,P_{\hat{S}})- n C(\eps ; W, P_{\hat{S}}) }{\sqrt{n}\, \Phi^{-1}(\eps)}\right)^2.
\end{align*}

\section{Results for the Capacity and Strong Converse} \label{sec:cap}
In this section, we state our main results for the two first-order quantities of interest, namely the $(\eps,\Gamma)$-capacity and the optimistic  $(\eps,\Gamma)$-capacity.  We also present conditions for the strong converse property  to hold. 

\subsection{The $(\eps,\Gamma)$-Capacity and the Optimistic  $(\eps,\Gamma)$-Capacity}
In order to state the $(\eps,\Gamma)$-capacity and the optimistic  $(\eps,\Gamma)$-capacity, we define  the following quantities:
\begin{align*}
J(R|\Gamma; W,P_{\hat{S}}) := \limsup_{n\to\infty} \Pr\bigg[ \frac{1}{n}\sum_{k=1}^n C_{S_k}(\Gamma) \leq R\bigg], \quad\mbox{and}\quad J^\dagger(R|\Gamma; W,P_{\hat{S}}) := \liminf_{n\to\infty} \Pr \bigg[ \frac{1}{n}\sum_{k=1}^n C_{S_k}(\Gamma) \leq R \bigg].
\end{align*}
where  the probabilities are taken with respect to the general state sequence $\hat{S}$. Note that $\Pr [ \frac{1}{n}\sum_{k=1}^n C_{S_k}(\Gamma) \leq R]$ is the cdf of the random variable $\frac{1}{n}\sum_{k=1}^n C_{S_k}(\Gamma) = C(T_{S^n}, \Gamma)$,
so $J(\cdot|\Gamma;W, P_{\hat{S}})$ and $J^\dagger(\cdot|\Gamma;W, P_{\hat{S}})$ are the $\limsup$ and $\liminf$ of this cdf, respectively.
\begin{theorem}[Capacity and Optimistic Capacity]  \label{thm:eps_cap}
For every $\eps \in [0,1]$, 
\begin{align}
C(\eps,\Gamma;W,P_{\hat{S}}) &= \sup\big\{R \, |\,  J(R|\Gamma;W,P_{\hat{S}})\le\eps \big\} , \label{eqn:thm_cap}
\quad \textrm{and}\\
C^\dagger(\eps,\Gamma;W,P_{\hat{S}}) &=\sup\big\{R \, |\, J^\dagger(R|\Gamma; W,P_{\hat{S}})<\eps \big\}. \label{eqn:thm_cap_dual}
\end{align}
\end{theorem}
Theorem~\ref{thm:eps_cap} is proved in Section~\ref{sec:prf_eps_cap}.  The case of most interest is the capacity-cost function when $\eps=0$ and the optimistic capacity-cost function when $\eps = 1$. In this case, it is easy to check from the definition of $J(R|\Gamma; W, P_{\hat{S}})$ and the $\pliminf$ that \eqref{eqn:thm_cap} reduces to 
\begin{align}
C(\Gamma; W,P_{\hat{S}})
=\pliminf_{n\to\infty} \frac{1}{n}\sum_{k=1}^n C_{S_k}(\Gamma) .\label{eqn:cap0}
\end{align}
Similarly, it is easy to verify that for $\eps=1$, the optimistic capacity in~\eqref{eqn:thm_cap_dual} reduces to 
\begin{align}
C^\dagger(\Gamma; W,P_{\hat{S}})= \plimsup_{n\to\infty} \frac{1}{n}\sum_{k=1}^n C_{S_k}(\Gamma) \label{eqn:cap1}
\end{align}
and, thus, $C(\Gamma; W, P_{\hat{S}}) \leq C^{\dagger}(\Gamma; W, P_{\hat{S}})$.
Note that both the $\eps$-capacity-cost function and its optimistic  version are expressed solely in terms of the sequence of random variables $\frac{1}{n}\sum_{k=1}^n C_{S_k}(\Gamma)$. This is because the channel $W^n$ is well-behaved; it is a memoryless and stationary channel from $\cX^n\times\cS^n$ to $\cY^n$ and thus has an effective first-order characterization in terms of  the capacity-cost functions $C_s(\Gamma), s\in\cS$. However, the state process $\hat{S}$ is general so, naturally, from information spectrum analysis~\cite{Han10, chen99}, we need to  invoke  probabilistic limit operations for the limiting cases of $\eps=0$ for the capacity and $\eps=1$ for the optimistic capacity. 

\begin{remark} \label{cor:nostate}
The converse bound in Theorem~\ref{thm:eps_cap} applies if no state information is present at the encoder and/or decoder. Furthermore, if state information is present \emph{only at the decoder}~\cite{ingber10} and if additionally 
$$\bigcap_{s \in \cS} \argmax_{P \in \cP(\cX)} I(P, W_s)$$ is non-empty, then the direct bounds in Theorem~\ref{thm:eps_cap} apply. The condition that the set above is non-empty is equivalent to saying that there exists a capacity-achieving input distribution  for each channel $P_{X|S}^*(\cdot|s)$ that is  identical across channels indexed by $s\in\cS$.
\end{remark}

The latter statement of the remark can be verified by inspecting the proof of Theorem~\ref{thm:eps_cap}. Simply note that capacity is achieved using Shannon-type i.i.d.\  random codes distributed according to the law $P_{X|S}$, where we choose $$P_{X|S}(\cdot|s) \in \argmax_{P \in \cP(\cX)} I(P, W_s)$$ for all $s \in \cS$.
If no side-information is available at the encoder, the corresponding one-shot bound still allows to choose a law $P_X \in \cP(\cX)$, independent of $s \in \cS$. This is not restrictive if $\bigcap_{s \in \cS} \argmax_{P \in \cP(\cX)} I(P, W_s)$ is non-empty.

\subsection{Strong Converse}
Uniting \eqref{eqn:cap0} and \eqref{eqn:cap1}  and recalling the definition of the strong converse property (as stated in~\cite{hayashi03}), we  immediately obtain the following:
\begin{corollary}[Strong Converse]\label{thm:str_conv}
A necessary and sufficient condition for the strong converse property to hold is 
\begin{align}
\pliminf_{n\to\infty} \frac{1}{n}\sum_{k=1}^n C_{S_k}(\Gamma) =\plimsup_{n\to\infty} \frac{1}{n}\sum_{k=1}^n C_{S_k}(\Gamma).\label{eqn:liminf_sup} 
\end{align}
\end{corollary}
In other words, for the strong converse to hold, the sequence of random variables $\frac{1}{n}\sum_{k=1}^n C_{S_k}(\Gamma)$ must converge (pointwise in probability) to $C(\Gamma) \ge 0$. Furthermore, when the strong converse holds, $C(\Gamma)$ is the capacity-cost function  which coincides with the optimistic capacity-cost function of the channel $W$ with general state $\hat{S}$. This condition is analogous to \cite[Thm.~1.5.1]{Han10} for almost-lossless source coding and   \cite[Thm.~3.5.1]{Han10} for channel coding.

While Corollary~\ref{thm:str_conv} provides a necessary and sufficient condition for the strong converse to hold, it requires the full statistics of  $\hat{S}$. Thus \eqref{eqn:liminf_sup} may be hard to verify in practice. We provide a simpler condition for the strong converse to hold that is based only on first- and second-order statistics.
\begin{corollary}[Sufficient Condition for Strong Converse]\label{cor:suff_str_conv}
The strong converse holds with capacity-cost function $C(\Gamma) \geq 0$ if the following limit  exists
\begin{align}
&\lim_{n \to \infty} \frac{1}{n} \sum_{k=1}^n\Exp[C_{S_k}(\Gamma)] = C(\Gamma) \quad \textrm{and} \label{eqn:limt_cond}\\
&\lim_{n \to \infty} \frac{1}{n^2} \sum_{k=1}^n\sum_{l=1}^n \Cov\left[C_{S_k}(\Gamma),C_{S_l}(\Gamma)\right] = 0 \label{eqn:cov_cond} 
\end{align}
\end{corollary}
It is easy to verify that~\eqref{eqn:limt_cond} and~\eqref{eqn:cov_cond} imply~\eqref{eqn:liminf_sup} using Chebyshev's inequality.
Observe that if the general source $\hat{S}$ decorrelates quickly such that  $\Cov\left[C_{S_k}(\Gamma),C_{S_l}(\Gamma)\right]$ is small for large lags $|k-l|$, then the covariance condition in~\eqref{eqn:cov_cond} is likely to be satisfied. In the following subsection, we provide some examples for which the covariance condition  either holds or is violated.

\subsection{Examples} \label{sec:examples}
In this section, we provide five examples to illustrate the generality of the model and the strong converse conditions. We assume that there are no cost constraints here.

\begin{example}[Mixed channels] \label{eg:mixed}
Let $S \leftarrow Q_S \in \cP(\cS)$ for $\cS = \{1, 2, \ldots, m\}$ and $C_1 < C_2 < \ldots < C_m$. Suppose the general source (state sequence) $\hat{S} := \{S^n = (S^{(n)}_1,\ldots, S^{(n)}_n)\}_{n=1}^{\infty}$ is such that $S^{(n)}_j= S$ for all $n\in\mathbb{N}$ and all  $1\le j\le n$.  Then, each covariance in \eqref{eqn:cov_cond} is equal to $\Var_{S\leftarrow Q_S}[C_{S}]$. If this variance is positive, then neither the sufficient condition in Proposition~\ref{cor:suff_str_conv} nor the necessary condition in~\eqref{eqn:liminf_sup} is satisfied. For such a state sequence, which corresponds to a {\em mixed channel}~\cite[Sec.~3.3]{Han10}, the $\eps$-capacity and optimistic $\eps$-capacity are given by
\begin{align*}
  C(\eps; W, P_{\hat{S}}) = \sup \{ R \, |\, \Pr [ C_S \leq R ] \leq \eps \} \quad \textrm{and} \quad C^{\dagger}(\eps; W, P_{\hat{S}}) = \inf \{ R \, |\, \Pr [ C_S \leq R ] \geq \eps \}.
\end{align*}
The two capacities coincide except at values of $\eps$ where $\sum_{s = 1}^{\ell} Q_S(s) = \eps$ for some $1\leq \ell \leq m$. There, the capacities are discontinuous and $C^{\dagger}(\eps; W, P_{\hat{S}}) = C_{\ell-1} < C(\eps; W, P_{\hat{S}}) = C_{\ell}$.\footnote{Note that $C(\eps; W, P_{\hat{S}})$ and  $C^{\dagger}(\eps; W, P_{\hat{S}})$ are  upper semi-continuous and lower semi-continuous in $\eps$ respectively.}
\end{example}

\begin{example}[State is i.i.d.]\label{eg:iid}
Suppose the source $\hat{S}$ is i.i.d.\ with common distribution $ \pi \in \cP(\cS)$. Then, $\Cov\left[C_{S_k},C_{S_l}\right]=0$ for $k\ne l$ and hence the  double sum in~\eqref{eqn:cov_cond} is simply $n \Var[C_S]$. This grows linearly in $n$ and hence, the strong converse condition holds with
\begin{align}
C(W, P_{\hat{S}}) =C(\pi)=\max_{P\in\cP(\cX|\cS)}I(P,W|\pi).  \label{eqn:capacity_iid}
\end{align} 
This recovers an elementary  and classical result by Wolfowitz~\cite[Thm.~4.6.1]{Wolfowitz}. Also see~\cite[Sec.~7.4.1]{elgamal} where the direct part is proved using multiplexing as in~\cite{Gold97}. 
\end{example}

\begin{example}[State is block i.i.d.] \label{eg:iid_mixed}
We now consider a mixture of the preceding two examples. Let $\nu \in  (0,1]$ and $\hat{S} := \{S^n = (S^{(n)}_1,\ldots, S^{(n)}_n)\}_{n=1}^{\infty}$. Define $d := \lfloor n^\nu \rfloor$ and write $n = m d+r$ where $0\le r< m$.  Partition the state sequence $S^n$ into $d+1$ subblocks where the first $d$ subblocks  are  each of length $m$ and the final subblock is of length $r$. Each subblock is independently assigned a state from some  common distribution $ \pi \in \cP(\cS)$ and that state is constant within each subblock. Thus, within each subblock (most of which have lengths  $\Theta(n^{1-\nu})$), the channel resembles a mixed channel (cf.~Example~\ref{eg:mixed}), while across the $\Theta(n^\nu)$ subblocks the channel evolves independently (cf.~Example~\ref{eg:iid}). Now the limit in \eqref{eqn:limt_cond} exists and the double sum of covariances in \eqref{eqn:cov_cond} can be computed to be $ (m^2 d + r^2  )\, \Var[C_S]$ which is of the order $\Theta( n ^{2-\nu})$. Since $\nu >0$, the covariance condition in~\eqref{eqn:cov_cond}  is satisfied and hence the strong converse condition holds with $C(W, P_{\hat{S}})=C(\pi)$ as in \eqref{eqn:capacity_iid}. Note that  Example~\ref{eg:mixed} corresponds to the limiting case $\nu\to 0$ and Example \ref{eg:iid} corresponds to the case $\nu=1$. 

\end{example}

\begin{example}[State is Markov]  \label{eg:markov}
Suppose that  the source  $\hat{S}$ evolves according to a time-homogenous, irreducible and ergodic (i.e.\ aperiodic and positive recurrent) Markov chain. Such a Markov chain admits a unique stationary distribution $\pi \in \cP(\cS)$. It is easy to check that $\abs{\Cov\left[C_{S_k},C_{S_l}\right]} \leq ae^{- b |k-l|}$ for some $a,b>0$ (see, e.g.~\cite[Thm.~3.1]{Bradley}). Also, $\frac{1}{n}\sum_{k=1}^n\Exp[C_{S_k}]\to\sum_s \pi(s) C_s$ regardless of the initial distribution. Hence, both~\eqref{eqn:limt_cond} and \eqref{eqn:cov_cond} are satisfied and this channel with Markov states admits a strong converse with  $C(W,P_{\hat{S}})=C(\pi)$ as in \eqref{eqn:capacity_iid}.

A variation of the Gilbert-Elliott channel~\cite{Gilbert, Elliott, Mush} with state information at the encoder and decoder is modeled in this way. In fact, since the above capacity can be achieved even without state information at the encoder (cf.\ Remark~\ref{cor:nostate}), our result recovers the strong converse for the regular Gilbert-Elliott channel.
\end{example}

\begin{example}[State is memoryless but non-stationary] \label{eg:weird}
The covariance condition~\eqref{eqn:cov_cond} is not sufficient in general. Consider the memoryless (but non-stationary) source $\hat{S}:= \{S^n = (S_1,\ldots, S_n)\}_{n=1}^{\infty}$ given by 
\begin{align}
S_k = \left\{ 
\begin{array}{cc}
S_{\mathrm{a}}& k\in\mathcal{J}\\
S_{\mathrm{b}} & k\notin\mathcal{J},
\end{array}
\right. \label{eqn:state_J}
\end{align}
where $\mathcal{J} :=\{ i \in\mathbb{N}: 2^{2k-1}\le i < 2^{2k}, k\in \mathbb{N}\}$.  This source is inspired by Example~3.2.3 in~\cite{Han10}. Since $\hat{S}$  is memoryless, just as in Example~\ref{eg:iid}, $\Cov\left[C_{S_k},C_{S_l}\right]=0$ for $k\ne l$. Hence, the covariance  condition is satisfied since the double sum scales as $\Theta(n)$. However, 
\begin{align*}
C(W,P_{\hat{S}})  =\pliminf_{n\to\infty} \frac{1}{n}\sum_{k=1}^n C_{S_k} =\frac{2c}{3}+\frac{d}{3},\quad\mbox{and} \quad 
C^\dagger(W,P_{\hat{S}}) =\plimsup_{n\to\infty} \frac{1}{n}\sum_{k=1}^n C_{S_k}  =\frac{ c}{3}+\frac{2d}{3}, 
\end{align*}
where the parameters $c:=\min\{\Exp[C_{S_{\mathrm{a}}}], \Exp[C_{S_{\mathrm{b}}}]\}$ and $d:=\max\{\Exp[C_{S_{\mathrm{a}}}], \Exp[C_{S_{\mathrm{b}}}]\}$.  If $ \Exp[C_{S_{\mathrm{a}}}]\ne  \Exp[C_{S_{\mathrm{b}}}]$, then $c<d$ and hence the necessary and sufficient condition in Corollary~\ref{thm:str_conv} is not satisfied and the strong converse property does not hold.   However, it can be verified  that the $\eps$-capacity for   $\eps \in [0,1)$ and optimistic $\eps$-capacity  for   $\eps \in (0,1]$ are equal to $C(W,P_{\hat{S}})$ and $C^\dagger(W,P_{\hat{S}}) $ respectively. This is a channel in which the strong converse  property does not hold but the $\eps$-capacity does not depend on $\eps\in [0,1)$ (cf.~\cite[Rmk.~3.5.1]{Han10}). 
\end{example}

\section{Results for the Second-Order Coding Rate  and  the Dispersion}\label{sec:second}

In this section, we state our main results for the second-order coding rates. We start by allowing the state sequence $\hat{S}:= \{S^n = (S^{(n)}_1,\ldots, S^{(n)}_n)\}_{n=1}^{\infty}$ to be general but discrete (as before). Subsequently, we revisit Examples~\ref{eg:mixed}--\ref{eg:weird} and derive explicit expressions for the  second-order coding rates and dispersions for these channels with states. In this section, in order not to complicate the exposition, we assume that there are no cost constraints, we use $\Ceps$ to denote $C(\eps; W, P_{\hat{S}})$, the capacity-achieving input distribution for each channel $W_s$ is unique\footnote{Let us further comment on the assumption that the  capacity-achieving input distribution for each channel $W_s$ is unique. This means that the the  set $\Pi(W_s):=\big\{ P\in\cP(\cX) \, \big|\, I(P,W_s)=C(W_s)\big\}$ is a singleton. In other words, the function $P\mapsto I(P,W_s)$ is strictly concave and so admits a unique maximum. This is easily seen to be satisfied for {\em almost all} (in a measure-theoretic sense) DMCs. More precisely, if the entries of the matrix $\{ W_s(y|x): x\in\cX,y\in\cY\}$ are randomly generated from a continuous probability distribution and then normalized so that $W_s(\cdot|x)$ sums up to one, then $\Pi(W_s)$ is a singleton and $P_s^*$ is unique almost surely. In addition,  the information dispersion $V_s$ is positive.  Canonical DMCs like binary symmetric channels, binary erasure channels, the z-channel all have unique   capacity-achieving input distributions. All input symmetric DMCs have  unique capacity-achieving input distributions that are uniform over $\cX$. In fact, all that we require in the following is that $\min\{V(P,W_s)\,|\, P\in\Pi(W_s)\}=\max\{V(P,W_s)\,|\, P\in\Pi(W_s)\}$ for all $s\in\cS$, which is  true if all the sets of  capacity-achieving input distributions $\Pi(W_s), s\in\cS$ are singletons. }
 and denoted by $P_s^*$, and  the information dispersions of all channels  are positive (i.e.,  $V_{\min}:=\min_{s\in\cS} V_s>0$).%
  
\subsection{The $(\eps, \beta)$-Optimum Second-Order Coding Rate for Channels with General States}
Let $T_{S^n} \in\cP_n(\cS)$ be the   type of the random length-$n$ state sequence $S^n=(S_1,\ldots, S_n)$. In order to state the $(\eps, \beta)$-optimum second-order coding rate in full generality, it is convenient to first define  the quantity
\begin{align}
K(r | R, \beta ; W, P_{\hat{S}}) := \limsup_{n\to\infty} \Exp\left[ \Phi\left( \frac{nR + n^\beta r -n C(T_{S^n})}{\sqrt{n V(T_{S^n})}}\right) \right],\label{eqn:def_expect_T}
\end{align}
where recall that $C(T_{S^n}) :=\sum_s T_{S^n}(s) C_s = \frac{1}{n} \sum_{k=1}^n C_{S_k}$ and $V(T_{S^n}) := \sum_s T_{S^n}(s) V_s = \frac{1}{n} \sum_{k=1}^n V_{S_k}$ are the average capacity and average information dispersion  with respect to the type $T_{S^n}$.  Note that the expectation in~\eqref{eqn:def_expect_T} is with respect to the  random type $T_{S^n}$ and since $V_{\min}>0$, the denominator  is positive.  The quantity $K(r | R,\beta ; W, P_{\hat{S}})$ plays a role that is similar to that played by $J(R|W,P_{\hat{S}})$ for the characterization of the $\eps$-capacity in Theorem~\ref{thm:eps_cap}. 

\begin{theorem}[General Second-Order Coding Rate] \label{thm:sec_gen}
For every $(\eps ,\beta)\in (0,1)\times[\frac{1}{2}, 1)$, 
\begin{equation}
\Lambda( \eps,\beta; W, P_{\hat{S}}) = \sup\left\{ r \, \big|\,  K(r |\, \Ceps,\beta ; W, P_{\hat{S}})\le\eps\right\}.\nonumber
\end{equation}
\end{theorem}
The proof of Theorem~\ref{thm:sec_gen}  is provided in Section~\ref{sec:prf_second}. The direct part uses a state-dependent Feinstein lemma~\cite{feinstein, Lan06} and the converse uses a recently-developed  converse technique by the  present authors~\cite{tomamicheltan12}. 

Theorem~\ref{thm:sec_gen}, which holds for any general discrete state distribution, can be interpreted as follows: Fix some blocklength $n$. If the type of $S^n = (S_1, \ldots, S_n)$  is known to be  some $t \in \cP_n(\cS)$ and we code at a rate of $C_\eps   + n^{\beta -1}r$, then it can be shown by the central limit theorem or its more quantitative variants such as the Berry-Esseen theorem~\cite[Sec.\ XVI.7]{feller} that the resultant optimal error probability is approximately 
\begin{equation*}
\Phi\left( \frac{n C_\eps    + n^\beta r -n C(t)}{\sqrt{n V(t)}}\right).
\end{equation*}
Since the state type is random, we expect that the optimal error probability as $n$ becomes large is given by $K(r | \Ceps, \beta ; W, P_{\hat{S}}) $, the limit superior of the expectation of   error probabilities conditioned on various  state types. 

\subsection{Specializations to    Various State Models}
While Theorem~\ref{thm:sec_gen} is a complete characterization of the $(\eps, \beta)$-optimum second-order coding rate, it is not computable in general as the discrete state sequence is allowed to be arbitrary. For example, it can be non-memoryless, non-stationary or non-ergodic. Thus, it is insightful to specialize this general result to various more tractable scenarios such as mixed channels, i.i.d.\ states,   and Markov states (cf.~Examples~\ref{eg:mixed}--\ref{eg:weird} in Section~\ref{sec:examples}).   The proofs of Theorems~\ref{thm:mixed}--\ref{thm:weird} are provided in Section~\ref{sec:prf_examples}.

\setcounter{example}{0}
\begin{example}[Mixed channels, continued]    Recall that here, the state random variable $S$ is identical for all  $n$. Let the state distribution be $Q_S \in \cP(\cS)$.  For simplicity, we suppose that the state  can only take on $2$ values, e.g., $\cS=\{\mathrm{a},\mathrm{b}\}$. Furthermore let $Q_S(\mathrm{a})=\alpha$ and $Q_S(\mathrm{b})=1-\alpha$ for some $\alpha \in (0,1)$. As usual, let the non-zero capacities and information dispersions of the channels $W_s$ be denoted as $C_s$ and $V_s$ respectively for $s \in \cS$. Without loss of generality, assume $C_{\mathrm{a}}\le C_{\mathrm{b}}$. Consider the following three cases:
\begin{enumerate}
\item Case I: $C_{\mathrm{a}} = C_{\mathrm{b}}$ 
\item Case II: $C_{\mathrm{a}} <  C_{\mathrm{b}}$  and $\eps< \alpha$
\item Case III: $C_{\mathrm{a}} <  C_{\mathrm{b}}$  and $\eps \ge \alpha$
\end{enumerate}
Note that in Case I, the $\eps$-capacity is the common value of the capacities $C_{\mathrm{a}}=C_{\mathrm{b}}$; in Case II, the $\eps$-capacity is $C_{\mathrm{a}}$ and finally in Case III, the $\eps$-capacity is $C_{\mathrm{b}}$ \cite[Ex.~3.4.2]{Han10}.  The specialization of Theorem~\ref{thm:sec_gen} yields a result by Polyanskiy-Poor-Verd\'u\cite[Thm.~7]{PPV11} which is stated as follows:
\begin{theorem}\label{thm:mixed}
The $(\eps, \frac{1}{2})$-optimum second-order coding rate of the mixed  channel  with two states is given as follows:
\begin{enumerate}
\item Case I: $\Lambda( \eps, \frac{1}{2}; W, P_{\hat{S}})$ is given as the solution $\Lambda$ to the following equation:
\begin{align*}
\alpha\, \Phi\left( \frac{\Lambda}{\sqrt{V_{\mathrm{a}}}}\right)+(1-\alpha) \, \Phi\left( \frac{\Lambda}{\sqrt{V_{\mathrm{b}}}}\right)=\eps.
\end{align*}
\item Case II: $\Lambda( \eps,\frac{1}{2}; W, P_{\hat{S}})$ is given as  
\begin{align*}
\Lambda \bigg( \eps,\frac{1}{2}; W, P_{\hat{S}} \bigg)=\sqrt{V_{\mathrm{a}}}\, \Phi^{-1}\left(\frac{\eps }{\alpha}\right).
\end{align*}
\item Case III: $\Lambda( \eps,\frac{1}{2}; W, P_{\hat{S}})$ is given as  
\begin{align*}
\Lambda \bigg( \eps,\frac{1}{2}; W, P_{\hat{S}} \bigg)=\sqrt{V_{\mathrm{b}}}\, \Phi^{-1}\left(\frac{\eps -\alpha}{1-\alpha}\right).
\end{align*}
\end{enumerate}
\end{theorem}
When $\eps=\alpha$,    $\Lambda( \eps, \frac{1}{2}; W, P_{\hat{S}}) = -\infty$. This corresponds to the critical point where we can code for the channel $W_{\mathrm{b}}$ (with capacity $C_{\mathrm{b}} > C_{\mathrm{a}}$) only if we tolerate zero error on the channel $W_{\mathrm{b}}$.   

To get an intuitive feel of how Theorem~\ref{thm:sec_gen} specializes to Theorem~\ref{thm:mixed}, we note that for every blocklength $n$, the state type $T_{S^n}=(T_{S^n}(\mathrm{a}), T_{S^n}(\mathrm{b})) \in \cP_n(\cS)$ can only be one of two values: Either $T_{S^n}=(1,0)$ with probability $\alpha$ or  $T_{S^n}=(0,1)$ with probability $1-\alpha$.  Theorem \ref{thm:mixed} follows by simply expanding the expectation in  \eqref{eqn:def_expect_T} and leveraging on the values of the $\eps$-capacity for the various cases to simplify the resultant limits.
\end{example}

\begin{example}[State is i.i.d., continued] 
In this example, the state sequence is i.i.d.\ with common distribution $\pi \in \cP(\cS)$. As we have seen, the strong converse property holds with the capacity given in \eqref{eqn:capacity_iid}. Define 
\begin{equation} \label{eqn:V_star}
V^*(\pi):= \Var_{S\leftarrow\pi}\left[C_S\right]
\end{equation}
to be the variance of the capacities of the constituent channels computed with respect to $\pi$. Recall also that $V(\pi):=\sum_s \pi(s)V_s > 0$ is the average of the information variances of the constituent channels. 
\begin{theorem} \label{thm:iid}
The $(\eps, \frac{1}{2})$-optimum second-order coding rate of the    channel with i.i.d.\ states is 
\begin{equation*}
\Lambda\bigg( \eps, \frac{1}{2}; W, P_{\hat{S}}\bigg) = \sqrt{\Upsilon(\eps; W , P_{\hat{S}})}\Phi^{-1}\left(\eps\right), \quad \textrm{where} \quad \Upsilon(\eps; W , P_{\hat{S}}) = V(\pi) + V^*(\pi) 
\end{equation*}
for all $\eps\in (0,1) \setminus\{ \frac{1}{2}\}$. More precisely,
\begin{align*}
  \log M^*(\eps; W^n, P_{S^n}) = n C(\pi) + \sqrt{n (V(\pi) + V^*(\pi))} \Phi^{-1}(\eps) + O(\log n) .
\end{align*}
\end{theorem}
New technical tools are required to prove Theorem \ref{thm:iid}. In particular, we apply a weak form of the Berry-Esseen theorem \cite[Thm.~2.2.14]{tao12} to the expectation in $K(r \, |\, \Ceps, \beta ; W, P_{\hat{S}})$ by taking into account the statistics of   $T_{S^n}$. Thus, we are using Berry-Esseen twice; once to account for the randomness of the channel given the state type resulting in $K(r |\, C_{\eps},\beta ; W, P_{\hat{S}})$, and the second to account for the randomness of the state which gives rise to $V^*(\pi)$. The variances $V(\pi)$ and $V^*(\pi)$ are   summed using the fact that the convolution of the PDFs of two independent Gaussians is a Gaussian and the resultant means and variances are the sum of the constituent ones. 

We note that $\Upsilon(\eps; W , P_{\hat{S}})$ can also be expressed as     $\Var[ i(X;Y|S) ]$ where  $(S,X,Y)\leftarrow \pi\times P^*\times W$.  This follows from the law of total variance  and the fact that the unconditional information variance equals the conditional information variance  for all capacity-achieving input distributions~\cite[Lem.~62]{PPV10}. Thus,   $\Var[i(X;Y|S)]=V(\pi) + V^*(\pi)$.  Note that $V^*(\pi)$ and  $V (\pi)$ represent the stochasticity  of  state and the channel given the state respectively. 

The dispersion we obtain   may appear to be identical to that for the problem of the state being available {\em only} at the decoder~\cite[Thm.~3]{ingber10}. However, we note an important distinction  between the two problems. In~\cite{ingber10}, the problem essentially boils down to channel coding where the output is $(Y,S)$ (and $S$ is independent of the channel input $X$). Hence, the law of total variance readily applies. However, the problem we solve here is much more involved, especially the converse part since the encoder {\em also} has the state information. 
\end{example}

\begin{example}[State is block i.i.d., continued] 
Here, we recall that the   state sequence $S^n$ of length $n = md+r$ is partitioned into $d+1$ subblocks where the first $d=\lfloor n^\nu\rfloor$ subblocks are of length $m$ and the final one is of length $r$. Each subblock is assigned a state which is an independent sample from $\pi$. Here is where the generality in the definition of the $(\eps,\beta)$-optimum second-order coding rate comes into play.  Let $V^*(\pi)$ be as in  \eqref{eqn:V_star}.
\begin{theorem} \label{thm:subblock} 
Let $\nu \in (0,1)$. The $(\eps,1-\frac{\nu}{2})$-optimum second-order coding rate of the    channel with block i.i.d.\ states and $d=\lfloor n^\nu\rfloor$ blocks is 
\begin{equation*}
\Lambda\bigg( \eps, 1-\frac{\nu}{2} ; W, P_{\hat{S}}\bigg) = \sqrt{ V^*(\pi)}\Phi^{-1}\left(\eps\right) .
\end{equation*}
\end{theorem}
In fact, the proof  demonstrates that the logarithm of the size of largest codebook for this channel   scales as
\begin{align*}
\log M^*(\eps;W^n,P_{S^n})=nC(\pi)+ \sqrt{nV(\pi)+n^{2-\nu}V^*(\pi)}\Phi^{-1}(\eps) + o(\sqrt{n}).
\end{align*}
Hence, the dominant second-order term is the one involving $n^{2-\nu}V^*(\pi)$ since $\nu<1$. From this expression, we recover the i.i.d.\ case in Theorem~\ref{thm:iid} in which $\nu = 1$. The intuition behind the term  $n^{2-\nu}V^*(\pi)$ is the following. The variance of the sum of random variables $\Var[\sum_{k=1}^n C_{S_k}]$ can be alternatively written as   $\Var[\sum_{j=1}^{d+1} E_{j}]$ where $E_j$ for $1\le j \le d$ represents the sum of $m$ identical copies of  $C_{S}$ and $E_{d+1}$ represents the sum of $r$ identical copies of $C_{S}$. Since $E_j$ for $1\le j \le d+1$  are independent random variables, $\Var[\sum_{k=1}^n C_{S_k}]$  is   equal to $(m^2  d + r^2 ) V^*(\pi)$ which is of the order $n^{2-\nu}V^*(\pi)$. We see  that by varying  $\nu \in (0,1)$, we can construct channels with state for which the second-order term in the expansion of $\log M^*(\eps ;W^n,P_{S^n})$ scales as $n^\beta$ for arbitrary $\beta \in [\frac{1}{2},1)$.  In previous works on   second-order coding rates and dispersions (e.g., \cite{Hayashi08,Hayashi09,PPV10,PPV11,strassen,Kot97,tomamichel12}), the second-order term always scales as $\sqrt{n}$ (except for exotic channels and $\eps > \frac{1}{2}$~\cite[Thm.~48]{PPV10}).
\end{example}

\begin{example}[State is Markov, continued] 
We revisit the example in which the state is governed by an irreducible, ergodic Markov chain that is given by a time-homogeneous transition kernel $M$ of size $|\cS|\times |\cS|$. The strong converse property holds with the capacity given in \eqref{eqn:capacity_iid}. We define 
\begin{align*}
V^{**}(M):=\Var_{S \leftarrow \pi} \left[C_S\right]+ 2   \sum_{j=1}^{\infty} \Cov_{S_1 \leftarrow \pi} \left[C_{S_1},C_{S_{1 +j}} \right]  
\end{align*}
as the analogue of $V^*(\pi)$ for the i.i.d.\ setting.   Here $\pi$ is the (unique) stationary distribution of $M$ and $S_1 \rightarrow S_2 \rightarrow \ldots \rightarrow S_n$ forms a Markov chain governed by the transition kernel $M$. We assume that the chain is started from the stationary distribution for simplicity (i.e., $S_1\leftarrow\pi$). When the state sequence is i.i.d., $V^{**}(M)=V^*(\pi)$ because the covariance terms vanish. Intuitively, the covariance terms quantify the amount of mixing in the Markov chain.
\begin{theorem} \label{thm:markov}
The $(\eps, \frac{1}{2})$-optimum second-order coding rate of the channel with Markov states is 
\begin{equation*}
\Lambda\bigg( \eps, \frac{1}{2}; W, P_{\hat{S}}\bigg) = \sqrt{\Upsilon(\eps; W , P_{\hat{S}})}\Phi^{-1}\left(\eps\right), \quad \textrm{where} \quad
\Upsilon(\eps; W , P_{\hat{S}}) = V(\pi) +V^{**}(M)
\end{equation*}
for all $\eps\in (0,1) \setminus\{ \frac{1}{2}\}$.
More precisely,  
\begin{equation*}
\log M^*(\eps; W^n, P_{S^n}) = n C(\pi) + \sqrt{n (V(\pi) + V^{**}(M))} \Phi^{-1}(\eps) + O(\log n) .
\end{equation*}
\end{theorem}
The proof of this result only requires slight modifications from the i.i.d.\ case, requiring a more general concentration bound of the   type to the stationary distribution~\cite{Kont05} and a Berry-Esseen theorem for weakly-dependent processes such as Markov processes \cite{Tik}.  

In fact,  $V^{**}(M)$ can be simplified as we  show  in  Lemma~\ref{lm:markov-simplified} in  Appendix~\ref{app:markov}. In particular,  we consider the state evolving like that for a Gilbert-Elliott channel~\cite{Gilbert, Elliott, Mush}, where  $\cS = \{ 0, 1\}$ is binary and the transition kernel is given, for $0 < \tau < 1$, by 
\begin{align*}
  M = \left( \begin{array}{cc} 1-\tau & \tau \\ \tau & 1-\tau \end{array} \right) \qquad \textrm{and} \qquad V^{**}(M) = \frac{1-\tau}{4\tau}(C_0 - C_1)^2 
\end{align*}
is a simple closed-form expression. 

We note that the $\eps$-dispersion for the Gilbert-Elliott channel with state information only at the decoder is $V(\pi) +V^{**}(M)$ \cite[Thm.~4]{PPV11}.  What our results in Theorem~\ref{thm:markov} and Lemma~\ref{lm:markov-simplified} show  is that the $\eps$-dispersion cannot be improved even when the state information is available at the encoder. 

This is unsurprising in light of the discussion in Remark~\ref{cor:nostate} and the fact that the capacity achieving input distribution is uniform for all $s \in \cS$. More generally, if $P_{s_1}^* = P_{s_2}^*$ for all $s_1, s_2 \in \cS$ then Theorems~\ref{thm:iid} and~\ref{thm:markov} also hold if the state information is only available at the decoder by the same argument used to justify Remark~\ref{cor:nostate}.
\end{example}

\begin{example}[State is memoryless but non-stationary, continued]
Finally, we revisit the example in which the state is memoryless but non-stationary and specifically given by \eqref{eqn:state_J}. We consider the simplification  that $S_{\mathrm{a}}=0$ and $S_{\mathrm{b}}=1$ with probability one. Also, assume $C_0 <C_1$. Then,  the $\eps$-capacity equals $\frac{2 C_0 }{3} + \frac{C_1 }{ 3}$ for all $\eps \in [0,1)$.  In this case, application of Theorem~\ref{thm:sec_gen} yields:
\begin{theorem} \label{thm:weird}
The $(\eps, \frac{1}{2})$-optimum second-order coding rate of the    channel with   states in Example~\ref{eg:weird} is  
\begin{equation*}
\Lambda\bigg( \eps, \frac{1}{2}; W, P_{\hat{S}}\bigg) = \sqrt{ \Upsilon(\eps; W , P_{\hat{S}})}\Phi^{-1}\left(\eps\right),
\quad \textrm{where} \quad \Upsilon(\eps; W , P_{\hat{S}}) = \frac{2V_0 }{3}+\frac{V_1}{3}
\end{equation*}
for all $\eps\in (0,1) \setminus\{ \frac{1}{2}\}$.
\end{theorem}
The intuition  behind Theorem~\ref{thm:weird} is that, given the deterministic nature of $S_{\mathrm{a}}$ and $S_{\mathrm{b}}$, there is only one type that is active in the expectation defining $K(r|\, \Ceps, \beta; W, P_{\hat{S}})$ for every blocklength $n$. The subsequence that attains the $\limsup$ in  $K(r| \, \Ceps, \beta; W, P_{\hat{S}})$ corresponds to the type $(\frac{2}{3},\frac{1}{3})$. Thus, both the $\eps$-capacity and $\eps$-dispersion correspond to the $(\frac{2}{3},\frac{1}{3})$-linear combination of $(C_0,C_1)$ and $(V_0,V_1)$ respectively. For random  $S_{\mathrm{a}}$ and $S_{\mathrm{b}}$, the determination of optimum second-order coding rates is much more difficult because $\frac{1}{\sqrt{n}}\sum_{k=1}^n C_{S_k}$ does not converge in distribution to a Gaussian unlike in the preceding examples. 
\end{example}

\section{Statements and Proofs of One-Shot Bounds} \label{sec:one-shot}

For the proofs of the direct parts, we require a state-dependent generalization of Feinstein's bound~\cite{feinstein,  Lan06}.

\begin{proposition}[State-Dependent Feinstein Bound]
 \label{thm:feinstein-direct}
 Let $\Gamma > 0$ and let $P \in \cP(\cX|\cS)$ be any input distribution. Then, for any $\eta > 0$ and $m \in \mathbb{N}$, there exists a code $\cC = \{ \cM, e, d\}$ with $|\cM| = m$ such that
\begin{align*}
  p_{\max}(\cC; W, P_S) &\leq \Pr [ i(X;Y|S) \leq \log |\cM| + \eta ] + \exp(-\eta) + \Pr[ b_S(X) > \Gamma] .
 \end{align*}
\end{proposition}

\begin{proof}
Feinstein's theorem with cost constraints~\cite{Lan06} applied to each $s\in\cS$ states that to every $P_s \in \cP(\cX)$ and $m_s\in\mathbb{N}$, there exists a code $\cC_s = \{ \cM_s, e_s, d_s\}$ with $|\cM_s|=m_s$ for channel $W_s$ that satisfies
\begin{align*}
 p_{\max}(\cC; W_s)&\leq \Pr [ i(X;Y|s) \leq \log |\cM_s| + \eta ] + \exp(-\eta) + \Pr[ b_s(X) > \Gamma] ,
\end{align*}
where $ p_{\max}(\cC; W_s)$ is the maximum error probability for channel $W_s$ using the code $\cC_s$.  Now for all states $s\in\cS$, set $\cM_s=\cM$. Take the expectation of the preceding bound with respect to $S$ to get 
\begin{align*}
\Exp[ p_{\max}(\cC; W_S)]&\leq \Pr [ i(X;Y|S) \leq \log |\cM | + \eta ]  + \exp(-\eta) + \Pr[ b_S(X) > \Gamma] ,
\end{align*}
where $(S, X) \leftarrow P(x|s) P_S(s)$. Let $M':=d_S(Y)$ be the estimate of the message. Now we can lower bound the expectation as follows:
\begin{align*}
\Exp[ p_{\max}(\cC; W_S)]&	= \sum_{s\in\cS}P_S(s) \max_{m \in \cM}  \Pr[ M'\ne M|M=m,S=s] \\
&\ge\max_{m \in \cM}\sum_{s\in\cS}P_S(s) \Pr[ M'\ne M|M=m,S=s] = p_{\max}(\cC; W, P_S).
\end{align*}
This completes the proof. 
\end{proof}

We prove a generalization of our one-shot converse in~\cite[Prop.~6]{tomamicheltan12}, which is known to be tight in third-order for discrete memoryless channels with a judiciously chosen output distribution. 
This can be viewed as a state-dependent generalization of the meta-converse by Polyanskiy-Poor-Verd\'u~\cite[Thm.~28]{PPV10} (specifically, Eq.~\eqref{eq:yuri} in the proof) and the information spectrum converse by Verd\'u-Han~\cite[Thm.~4]{VH94},\cite[Lem.~3.2.2]{Han10} (specifically, Eq.~\eqref{eq:han} in the proof).

\begin{proposition}[State-Dependent Function Converse]
  \label{thm:one-shot-conv}
 Let $0 \leq \eps \leq 1$ and let $\Gamma > 0$. 
 Moreover, let $\{ \cS_t \}_{t \in \cT}$ be a partitioning of $\cS$ into mutually disjoint subsets and let $T_S$ be the variable indicating the random partition $S$ belongs to.
 Then, for any $\delta > 0$,
  \begin{align*}
    &\log M^*(\eps, \Gamma; W, P_{S}) \leq 
    \inf_{ Q \in \cP(\cY|\cS) }\ \sup_{ f:\, \cT \to \cX \times \cS}\ \sup \big\{ R \, \big| \Pr 
    \big[ j_Q(X;Y|S) \leq R \,\big|\, (X, S) = f(T_S) \big] \leq \eps + \delta \big\} - \log \delta.
  \end{align*}
  where $f: \cT \to \cX \times \cS$ is any function such that $f(t) = (x, s)$ satisfies $s \in \cS_t$ and $x \in \cB_s(\Gamma)$ 
  for all $t \in \cT$. 
\end{proposition}

\begin{remark}
For i.i.d.\ repetitions of the channel on blocks of length $n$, the natural partitioning is into type classes $\cS_t := \{s^n\in\cS^n \,|\, T_{s^n} = t\}$. Let  $\cT = \cP_{n}(\cS)$ be the set of all types   and let $T_{S^n}$  denote the random type of $S^n$. 
\end{remark}

For the proof, we will need the following quantity~\cite{wang09,WangRenner}. Let $\eps \in (0,1)$ and let $P, Q \in \cP(\cZ)$. We consider binary (probabilistic) hypothesis tests $\xi : \cZ \to [0,1]$ and define the \emph{$\eps$-hypothesis testing divergence}
\begin{align*}
 D_h^{\eps}(P \| Q) := \sup \bigg\{ R  \,\bigg|\, \exists\ \xi :\ 
   \sum_{z \in \cZ} Q(z) \xi(z) \leq (1-\eps) \exp(-R)\ \land\ \sum_{z \in \cZ} P(x) \xi(z) \geq 1-\eps  \bigg\}.
\end{align*}
Note that $D_h^{\eps}(P \| Q)=-\log \frac{\beta_{1-\eps}(P,Q)}{1-\eps}$  where $\beta_\alpha$ is defined in PPV~\cite[Eq.\ (100)]{PPV10}. It is easy to see that $D_h^{\eps}(P\|Q)\geq 0$, where the lower bound is achieved if and only if $P = Q$ and $D_h^{\eps}(P \| Q)$ diverges if $P$ and $Q$ are singular measures. It satisfies a data-processing inequality~\cite[Lem.~1]{wang09}
\begin{align*}
   D_h^{\eps}(P \| Q) \geq D_h^{\eps}(PW \| QW) 
   \qquad \textrm{for all channels $W$ from $\cZ$ to $\cZ'$} .
\end{align*}

\begin{proof}[Proof of Prop.~\ref{thm:one-shot-conv}]
  We consider a general code $\{ \cM, e, d\}$, where the encoder is such that
  $e(s,m) \in \mathcal{B}_s(\Gamma)$ for all $m \in \cM$, $s \in \cS$. Moreover, let $Q \in \cP(\cY|\cS)$ be arbitrary for the moment. 
  
  Due to the data-processing inequality for the $\eps$-hypothesis testing divergence and the relation between random variables in Fig.~\ref{fig:states}\footnote{Note, in particular, that $(M,S) \leftrightarrow (X, S) \leftrightarrow (Y, S) \leftrightarrow M'$ forms a Markov chain.}, we have
  \begin{align*}
    D_h^{\eps}\big(P_{XYS} \| Q_{XYS} \big) \geq 
       D_h^{\eps}\big(P_{MYS} \| P_{M} \times Q_{YS} \big) \geq  D_h^{\eps}\big(P_{MM'} \| P_M \times Q_M\big) .
  \end{align*}  
  Here, $P_{XYS} = W(y|x,s) P(x|s) P_S(s)$ is induced by the encoder applied to $S \leftarrow P_S$ and a uniform $M$. In particular, note that $P(x|s) = 0$ if $x \notin \mathcal{B}_s(\Gamma)$ and that the encoding operation $E(x,s'|m,s) = 1\{s \!=\! s'\} e(x|m,s)$ can be inverted probabilistically. In contrast, $Q_{XYS}$ is an alternative hypothesis of the form
  \begin{align*}
    Q_{XYS}(x,y,s) = Q(y|s) P(x|s) P_S(s),
  \end{align*}
  where the channel output does not depend on the channel input $x$, but does depend on its state $s$. After employing the argument in~\cite[Lem.~3]{wang09} and~\cite[Prop.~6]{tomamicheltan12} to find $D_h^{\eps}(P_{MM'} \| P_M \times Q_M) \geq \log |\cM| + \log (1-\eps)$, this directly yields a generalization of the meta-converse~\cite[Thm.~28]{PPV10} to channels with state
  \begin{align}
   \log M^*(\eps,\Gamma; W, P_S) \leq \sup_{P \in \cP(\cX|\cS)} \inf_{Q \in \cP(\cY|\cS)} D_h^{\eps}(P_{XYS} \| Q_{XYS} ) + \log \frac{1}{1-\eps} ,
    \label{eq:yuri}
  \end{align}
  where the maximization is only over input distributions $P \in \cP(\cX|\cS)$ that satisfy $P(x|s) = 0$ for all $x \notin \cB_s(\Gamma)$.
  
  Instead, we are interested in a relaxation of this bound following the lines of~\cite{tomamicheltan12}. Also see \cite[Eq.~(102)]{PPV10}.
  Let $\delta > 0$ be arbitrary. We can further upper-bound $M^*(\eps;\Gamma; W,P_{S})$ in terms of the information spectrum (see, e.g.~\cite[Lem.~2]{tomamicheltan12}) to find
  \begin{align}
    \log M^*(\eps,\Gamma; W, P_S) \leq \sup_{P \in \cP(\cX|\cS)} \inf_{Q \in \cP(\cY|\cS)} \sup \big\{ R \, \big| \Pr 
    \big[ j_Q(X;Y|S) \leq R \big] \leq \eps + \delta \big\} .
    - \log \delta \label{eq:han} 
  \end{align}
    We may expand $\Pr[ j_Q(X;Y|S) \leq R]$ as follows:
  \begin{align*}
    \Pr[ j_Q(X;Y|S) \leq R] &= \sum_{t \in \cT} \Pr [ T_S = t ] \Pr[ j_Q(X;Y|S) \leq R \,|\, T_S = t ]\\
    &= \sum_{t \in \cT} \Pr [ T_S = t ] 
    \sum_{s \in \cS_t}\, \sum_{x \in \cB_s(\Gamma)} \!\!\!
    \Pr [ X = x, S = s \,|\, T_S = t ]
      \Pr \big[ j_Q(X;Y|S) \leq R \,\big|\, X = x, S = s\big] .
  \end{align*}
Clearly, for every $t \in \cT$, there exist symbols $s_Q^*(t)$ and $x_Q^*(t)$ with
\begin{align*}
 \big( s_Q^*(t), x_Q^*(t) \big) \in \argmin_{ \{ (s, x) \in \cS_t \times \cX \,|\, x \in \cB_s(\Gamma) \}}
\Pr[ j_Q(X;Y|S) \leq R \,|\, X=x, S=s]
\end{align*}
such that
$\Pr[ j_Q(X;Y|S) \leq R \,|\, X = x_Q^*(t), S = s_Q^*(t) ] \leq
   \Pr[ j_Q(X;Y|S) \leq R \,|\, T_S = t ] .$
Hence, we can relax the condition on $R$ in the supremum to get
\begin{align*}
     &\sup \big\{ R \, \big| \Pr \big[ j_Q(X;Y|S) \leq R \big] \leq \eps + \delta \big\} \\
    &\qquad \leq \sup \big\{ R \,\big| \Pr \big[ j_Q(X;Y|S) \leq R \,\big|\, X = x_Q^*(T_S), S = s_Q^*(T_S) \big] \leq \eps+ \delta \big\}  -\log \delta \\
    &\qquad \leq \sup_{f:\, \cT \to \cX \times \cS}\ \sup \big\{ R \,\big| \Pr \big[ j_Q(X;Y|S) \leq R \,\big|\, (X, S) = f(T_S) \big] \leq \eps+ \delta \big\} 
    -\log \delta ,
  \end{align*}
  where the function $f$ is of the form described in the statement of the lemma. Note that the last expression is independent of the input distribution $P \in \cP(\cX|\cS)$ induced by the code. Thus, substituting into~\eqref{eq:han}, the outer supremum over $P$ can  be dropped, concluding the proof.
\end{proof}

\section{Proofs of First-Order Results and Strong Converse}
\label{sec:prf_eps_cap}

We ignore cost constraints in the remainder to make the presentation more concise. We note that for the direct part, we can handle cost constraints by using the concavity and hence continuity of $\Gamma\mapsto C(\eps, \Gamma; W, P_{\hat{S}})$ and $\Gamma\mapsto C^\dagger(\eps, \Gamma; W, P_{\hat{S}})$.  For the converse part, we simply restrict the set of inputs to those that are admissible. For the proof of Theorem~\ref{thm:eps_cap}, we consider direct and converse bounds separately.

We recall that the conditional and unconditional  information variances~\cite{PPV10} for a given input distribution $P \in\cP(\cX)$ and a discrete memoryless channel $W\in\cP(\cY|\cX)$ are respectively defined as  
\begin{align*}
V(P,W) & := \sum_x P(x) \sum_y \bigg[ \log\frac{W(y|x)}{PW(y)}- D(W(\cdot|x) \| PW) \bigg]^2,\qquad\mbox{and} \\
U(P,W) & := \sum_x P(x) \sum_y  \bigg[  \log\frac{W(y|x)}{PW(y)}-C\bigg]^2  ,
\end{align*}
where $C = \max_P I(P,W)$ is the capacity of the channel $W \in \cP(\cY|\cX)$.  In particular, uniting~\cite[Lem.~62]{PPV10} and~\cite[Rmk.~3.1.1]{Han10}, the following uniform bounds hold
\begin{align}
  V(P,W) \leq U(P,W) \leq V^+:= 2.3 |\cY|. \label{eq:uni}
\end{align}
Furthermore, $U(P,W)=V(P,W)$ for all $P$ satisfying $I(P,W)=C$.

In the proofs, we initially fix $n \in \mathbb{N}$ and consider the $n$-fold memoryless channel $W^{n}\in \cP( \cY^n|\cX^{n} \times \cS^{n} )$ and a state sequence $S^n$ governed by the distribution $P_{S^n} \in \cP(\cS^{n})$. 

The limiting cases $C(1; W, P_{\hat{S}}) = \infty$ and $C^{\dagger}(0; W, P_{\hat{S}}) = -\infty$ follow immediately from the definition and we exclude them in the following.

\subsection{Direct Part for $(\eps,\Gamma)$-Capacity and Optimistic $(\eps,\Gamma)$-Capacity}

In the following, we show that, for $\eps \in [0, 1]$,
\begin{align}
 C(\eps; W, P_{\hat{S}}) &\geq \sup\{R \, | \,  J(R|W, P_{\hat{S}}) \leq \eps \},  \qquad \textrm{for } \eps \in [0, 1), \qquad \textrm{and} \label{d1}\\
 C^{\dagger}(\eps; W, P_{\hat{S}}) &\geq \sup\{R \, |\,  J^{\dagger}(R|W, P_{\hat{S}}) < \eps \}, \qquad\hspace{-0.15cm} \textrm{for } \eps \in (0, 1] . \label{d2}
\end{align}

\begin{proof}
   Consider an input distribution $P^* \in \cP(\cX|\cS)$ that satisfies $P^*(\cdot|s) \in \argmax_{P \in \cP(\cX)} I(P, W_s)$. We now apply Proposition~\ref{thm:feinstein-direct} to $W^n$ with input distribution $P(x^n|s^n) = \prod_{k=1}^n P^*(x_k|s_k)$.  For  any $R > 0$, we find that there exists a code $\cC = \{\cM_n, e, d\}$, with $|\cM_n| =\lfloor \exp(n R)\rfloor$,  that satisfies
  \begin{align*}
    p_{\max}(\cC; W^n, P_{S^n}) \leq \underbrace{\Pr\bigg[ \frac{1}{n} i(X^n; Y^n| S^n) \leq R + \eta \bigg]}_{=:\, \textrm{p}(R)} \!+ \exp(-n \eta) ,
  \end{align*}
for any $\eta > 0$.  The probability above can be bounded as follows
  \begin{align}
    \textrm{p}(R) &= \sum_{s^n \in \cS^n} \Pr[S^n = s^n]
    \Pr \bigg[ \frac{1}{n} \sum_{k=1}^n i(X_k; Y_k| s_k) \leq R + \eta \,\bigg|\, S^n = s^n \bigg] \nonumber\\
    &\leq \sum_{s^n \in \cS^n} \Pr[S^n = s^n] \Bigg( 1 \bigg\{ \frac{1}{n} \sum_{k=1}^n C_{s_k} \leq R + 2 \eta\bigg\}  + \Pr \bigg[ \frac{1}{n} \sum_{k=1}^n i(X_k; Y_k| s_k) < \frac{1}{n} \sum_{k=1}^n C_{s_k} - \eta \,\bigg|\, S^n = s^n \bigg] \Bigg) \nonumber\\
    &\leq \Pr  \bigg[ \frac{1}{n} \sum_{k=1}^n C_{S_k} \leq R + 2 \eta\bigg] + 
    \frac{V^+}{n \eta^2} , \label{d-all}
  \end{align}
  where we employed Chebyshev's inequality and~\eqref{eq:uni}.
  
  To show~\eqref{d1}, set $R^* := \sup \{ R | J(R|W,P_{\hat{S}}) \leq \eps\} - 3\eta$ and note that
  $R^*$ is $\eps$-achievable since~\eqref{d-all} implies
  \begin{align}
    &\limsup_{n \to \infty} p_{\max}(\cC; W^n, P_{S^n}) \leq \limsup_{n \to \infty} \Pr\! \bigg[ \frac{1}{n} \sum_{k=1}^n C_{S_k} \leq R + 2 \eta\bigg] \leq \eps , \label{eq:d-end}
  \end{align}
   where the last inequality is by definition of $J(R|W,P_{\hat{S}})$. Since this holds for all $\eta > 0$,
   \eqref{d1} follows. Equation~\eqref{d2} follows analogously by defining $R^* := \sup \{ R | J^{\dagger}(R|W,P_{\hat{S}}) < \eps\} - 3\eta$   and taking a $\liminf$ in~\eqref{eq:d-end}.
\end{proof}

\subsection{Converse Part for $(\eps,\Gamma)$-Capacity and Optimistic $(\eps,\Gamma)$-Capacity} 

In the following, we show that, 
\begin{align}
 C(\eps; W, P_{\hat{S}}) &\leq \sup\{R  \, |\,  J(R|W, P_{\hat{S}}) \leq \eps \}, \qquad \textrm{for } \eps \in [0,1), \qquad \textrm{and} \label{co1}\\
 C^{\dagger}(\eps; W, P_{\hat{S}}) &\leq \sup\{R \, |\,  J^{\dagger}(R|W, P_{\hat{S}}) < \eps \}, \qquad\hspace{-0.15cm} \textrm{for } \eps \in (0,1]. \label{co2}
\end{align}

\begin{proof}
For $n$ repetitions of the channel, we partition $\cS^{n}$ into type classes such that $\cT = \cP_n(\cS)$ and  $T_{S^n}$ denotes the random type of $S^n$. 
We define the \emph{joint type}~\cite{Csi97} of $x^n$ and $s^n$ as
\begin{align}
  T_{x^n,s^n}(x,s) := \frac{1}{n} \sum_{k=1}^n 1 \{s_k=s\} 1 \{x_k=x\} = T_{x^n|s^n}(x|s) T_{s^n}(s) \, . \label{eq:inducedtype}
\end{align}
Note that $\sum_x T_{x^n}(x,s) = T_{s^n}(s)$ is the type of $s^n$ and $T_{x^n,s^n}(x|s)$ is well-defined if we set $T_{x^n|s^n}(x|s)$ to uniform for all $s \in \cS$ such that $T_{s^n}(s) = 0$. We collect all these conditional types in a set $\cP_n(\cX|\cS)$ satisfying $\abs{\cP_n(\cX|\cS)} \leq (n+1)^{|\cX| |\cS|}$. We also define $T_{X^n,S^n}$ and $T_{X^n|S^n}$ as the corresponding random joint and conditional types.

We then apply Proposition~\ref{thm:one-shot-conv} for $\eps_n \in [0, 1]$, which yields
\begin{align*}
  &\log M^*(\eps_n; W^n, P_{S^n})  \leq\sup_{f:\, \cT \to \cX^n \times \cS^n}  \sup \big\{ R \, \big| \Pr 
    \big[ j_{Q^{(n)}}(X^n;Y^n|S^n) \leq R \,\big|\, (X^n, S^n) = f(T_{S^n}) \big] \leq \eps_n + \delta \big\} - \log \delta, 
\end{align*}
where we employed the following convex combination of conditional distributions:
\begin{align*}
Q^{(n)}(y^n|s^n) := \frac{1}{\abs{ \cP_n(\cX|\cS) }} \sum_{T_{x^n|s^n} \in \cP_n(\cX|\cS)}\, \prod_{k=1}^n T_{x^n|s^n}W(y_k|s_k) \,. 
\end{align*}
We also recall that $f$ is a function mapping $t \in \cP_n(\cS)$ to an element of $\big\{ (x^n, s^n) \in \cX^{n} \times \cS^{n} \,\big|\, T_{s^n} = t \big\}$.

The inner probability can be bounded as follows.
\begin{align*}
  &\Pr \big[ j_{Q^{(n)}}(X^n;Y^n|S^n) \leq R \,\big|\, (X^n,S^n) = f(T_{S^n}) \big] \\
  &\qquad  = \sum_{t \in \cT} \Pr[ T_{S^n} = t ] \Pr \Bigg[ \log \frac{\prod_{k=1}^n W(Y_k|X_k,S_k)}{Q^{(n)}(Y^n|S^n)} \leq R \,\Bigg|\, (X^n,S^n) = f(t) \Bigg] \\
  &\qquad  \geq \sum_{t \in \cT} \Pr[ T_{S^n} = t ] \Pr \Bigg[ \sum_{k=1}^n \log \frac{W(Y_k|X_k,S_k)}{ T_{X^n|S^n}W(Y_k|S_k)} \leq  R - \log \big|\cP_n(\cX|\cS)\big| \,\Bigg|\, (X^n,S^n) = f(t) \Bigg] .
\end{align*}
Here, we chose the conditional type $T_{X^n|S^n}$ depending on the joint type of $X^n$ and $S^n$.

Thus, choosing $\delta = 1/\sqrt{n}$ hereafter, we find the following bound:
\begin{align}
  &\qquad \frac{1}{n} \log M^*(\eps_n; W^n, P_{S^n}) \leq \sup_{f:\, \cT \to \cX^n \times \cS^n} \textrm{cv}(f; n,\eps_n) + \frac{\log n}{2n} + |\cX||\cS|\,\frac{\log (n + 1)}{n}, \label{eq:tobeusedagain}
\end{align}
and we introduced the quantities $\textrm{cv}(f;n,\eps_n) := \sup \big\{  R \,\big|\, \Exp \big[ \Xi_{f}(R; n, T_{S^n}) \big] \leq \eps_n + \delta \big\}$ and
\begin{align}
  \Xi_{f}(R; n, t) &:= \Pr \bigg[ \frac{1}{n} \sum_{k=1}^n j_{T_{X^n|S^n}}(X_k;Y_k|S_k) \leq R \,\Bigg|\, (X^n,S^n) = f(t) \Bigg] . \label{eqn:def_Xi}
\end{align}

We next analyze $\cv(f;n,\eps_n)$ for a fixed function $f$. First, note that $\Xi_f(R; n,t)$ is the cumulative distribution function of a sum of independent random variables since $S^n$ and $X^n$ are fixed. 
By Chebyshev's inequality, we find
\begin{align*}
  \Xi_{f}(R; n, t) \geq 1 - \frac{ V(T_{x^n|s^n}, W| t) }{n \big( R - I(T_{x^n|s^n}, W| t) \big)^2} \geq 1 - \frac{V^+}{n (R - C(t) )^2} \qquad \textrm{if} \ R > C(t) .
\end{align*}
Here, the expectation and variance are defined for $(x^n,s^n) = f(t)$ as
\begin{align*}
  \Exp \Bigg[ \frac{1}{n} \sum_{k=1}^n  j_{T_{X^n|S^n}}(X_k; Y_k|S_k) \,\bigg| \,(X^n,S^n) = f(t) \Bigg] &= I \big(T_{x^n|s^n}, W \big| t \big), \\
  \Var \Bigg[ \frac{1}{n} \sum_{k=1}^n j_{T_{X^n|S^n}}(X_k; Y_k|S_k)\,\bigg|\, (X^n,S^n) = f(t) \Bigg] &= \frac{V \big(T_{x^n|s^n}, W \big| t \big)}{n} ,
\end{align*}
and $V(P,W|t)=\sum_s t(s) V(P_s,W_s)$. The second inequality follows since
\begin{align*}
  I\big(T_{x^n|s^n}, W \big| t \big) &= \sum_{s \in \cS} t(s) I\big(T_{x^n|s^n}(\cdot|s), W_{s} \big) 
  \leq \sum_{s \in \cS} t(s) C_{s} = C(t) = \frac{1}{n} \sum_{k=1}^n C_{s_k} 
\end{align*}
and $V$ is bounded using~\eqref{eq:uni}.

We employ Markov's inequality which states that $\Exp [ \Xi_f(R; n, T_{S^n}) ] \geq \gamma \Pr [ \Xi_f(R; n, T_{S^n}) \geq \gamma ]$ for any $\gamma > 0$. Thus, choosing $\gamma = 1 - 1/\sqrt{n}$, we find
\begin{align}
  \cv(f; n, \eps_n) &\leq \sup \bigg\{ R\, \bigg|\, \Pr [ \Xi_f(R; n, T_{S^n}) \geq \gamma ] \leq \frac{\eps_n+\delta}{\gamma} \bigg\} 
  \nonumber\\
  &\leq \sup \bigg\{ R \, \bigg|\, \Pr \bigg[ \frac{1}{n} \sum_{k=1}^n C_{S_k} \leq R \bigg] \leq \frac{\eps_n+\delta}{\gamma} \bigg\} 
  + \sqrt{\frac{V^+}{\sqrt{n}}}\ , \label{a-bound}
\end{align}
which is independent of $f$.

Finally, the asymptotics~\eqref{co1} and~\eqref{co2} can be shown as follows. Due to the above, any $\eps$-achievable rate $R$ satisfies
\begin{align*}
  R &\leq \liminf_{n \to \infty} \frac{1}{n} \log M^*(\eps_n; W^n, P_{\hat{S}}) \leq \liminf_{n \to \infty} \sup \bigg\{ R \,\bigg|\,  \Pr \bigg[ \sum_{k=1}^n C_{S_k}\le R \bigg] 
  \leq \eps_n' \bigg\}
\end{align*}
for some sequence $\hat{\eps} = \{ \eps_n \}_{n=1}^\infty$ with $\limsup_{n \to \infty} \eps_n \leq \eps$
and $\hat{\eps}'$ defined via $\eps_n' = (\eps_n+\delta)/\gamma$ with $\delta = 1/n$ and $\gamma = 1 - 1/\sqrt{n}$ such that the limits of the sequences coincide. Hence, for any $\xi > 0$, there exists a constant $N_\xi \in \mathbb{N}$ such that
\begin{align*}
  \forall\,  n \geq N_\xi:\ \ \Pr \bigg[\frac{1}{n} \sum_{k=1}^n C_{S_k} \leq R - \xi \bigg] \leq \eps_n' .
\end{align*}
Thus, in particular,
\begin{align}
 J(R-\xi|\Gamma; W,P_{\hat{S}})= \limsup_{n\to\infty} \Pr \bigg[\frac{1}{n} \sum_{k=1}^n C_{S_k} \leq R - \xi \bigg] \leq 
  \limsup_{n\to\infty}\eps_n' \leq \eps \,. \label{eq:implied}
\end{align}
Now define 
\begin{align*}
R_0:=\sup\big\{ R \, \big| \, J(R|\Gamma; W,P_{\hat{S}}) \le\eps \big\},
\end{align*}
and suppose, to the contrary, that $R>R_0$. This means that we can find a $\xi>0$ such that $R-\xi> R_0$. Choose this $\xi$ for the above argument. Thus,  $R-\xi> \sup\big\{ R \, \big| \, J(R|\Gamma; W,P_{\hat{S}}) \le\eps \big\}$ which means that $J(R-\xi|\Gamma; W,P_{\hat{S}})>\eps$, contradicting \eqref{eq:implied}.  This means that we must have $R\le R_0= C(\eps,\Gamma;W,P_{\hat{S}})$.

The second statement~\eqref{co2} follows analogously by choosing a sequence~$\hat{\eps}$ with $\liminf_{n \to \infty} \eps_n < \eps$ and taking a $\liminf$ in~\eqref{eq:implied}.
\end{proof}

\section{Proof  of the General Second-Order Result in Theorem~\ref{thm:sec_gen}}
\label{sec:prf_second}

In this section, we prove Theorem~\ref{thm:sec_gen}. Before we start, let 
\begin{align*}
L(P,W):=\sum_x P(x) \sum_y W(y|x) \left|\log \frac{W(y|x)}{PW(y)} - D(W(\cdot|x ) \| PW) \right|^3
\end{align*}
be the third absolute moment of the log-likelihood ratio between $W(\cdot|x) , PW \in\cP(\cY)$. We will often leverage on the following auxiliary result which follows by a straightforward modification of~\cite[Lem.~46]{PPV10}.
%
\begin{lemma}[Uniform Bound on Third Moment] \label{lem:third}
If $|\cX|$ and $|\cY|$ are finite, 
\begin{align}
L(P,W)\le L^+:= |\cY| \big(9 e^{-1} \log e\big)^3 \label{eqn:boundL}
\end{align}
for all $P\in\cP(\cX)$ and $W\in\cP(\cY|\cX)$.
\end{lemma}
Also recall that $W$ satisfies $V_s \ge V_{\min}>0$ for all $s\in\cS$, which allows us to define a universal Berry-Esseen constant $B := 6 L^+/ V_{\min}^{3/2}$ that will be used frequently in the following.

 Theorem~\ref{thm:sec_gen} follows from the following two results   characterizing the channel's performance for block-length $n$ in terms of $C(T_{S^n})$  and $V(T_{S^n})$. 
\begin{proposition}[Direct Part] \label{lem:direct_second}
For each $n \in \mathbb{N}$, there exists a length-$n$ code $\cC_n = \{ \cM, e, d\}$ with
\begin{equation}
p_{\max}(\cC_n; W^n, P_{S^n}) \le\Exp  \left[ \Phi \left( \sqrt{n}\cdot\frac{R  - C(T_{S^n})  }{\sqrt{  V(T_{S^n}) }}\right) \right] + \frac{D_1 \log n}{\sqrt{n}} + \frac{B+1}{\sqrt{n}}  \label{eqn:lemma_direct}
\end{equation}
where   $R:=\frac{1}{n}\log |\cM|$ 
and $D_1 := (2 \sqrt{2\pi V_{\min}} )^{-1}$. 
\end{proposition}
\begin{proposition}[Converse Part] \label{lem:converse_second}
Let $0 < \eps < 1$. Then, for $n \geq N_0$ we have
\begin{align*}
  \frac{1}{n} \log M^*\big(\eps; W^n, P_{S^n}\big) \leq \sup \Bigg\{ R \,\Bigg|\, \Exp \Bigg[ \Phi\Bigg( \sqrt{n} \cdot \frac{R - C(T_{S^n})}{\sqrt{V(T_{S^n})}} \Bigg) \Bigg] \leq \eps +  \frac{B+1}{\sqrt{n}}\Bigg\} + \frac{D_2\log n}{n}  .
\end{align*}
where  $N_0$ and $D_2 > 0$   only depend on the parameters of $W$ (cf.~Lemmas~\ref{lm:Pi-mu} and~\ref{lm:psit}) and the cardinalities $|\cX|$, $|\cY|$ and $|\cS|$.
\end{proposition}
\begin{remark}
  Proposition~\ref{lem:converse_second} can be restated as follows. Every length-$n$ code $C_n = \{\cM, e, d\}$ for $n \geq N_0$ satisfies
  \begin{align*}
    p_{\mathrm{avg}}(\cC_n; W^n, P_{S^n}) \geq \Exp  \left[ \Phi \left( \sqrt{n}\cdot\frac{R  - C(T_{S^n})  }{\sqrt{  V(T_{S^n}) }}\right) \right] - \frac{D_3 \log n}{\sqrt{n}} - \frac{B+1}{\sqrt{n}} ,
  \end{align*}
  where   $R:=\frac{1}{n}\log |\cM|$ and $D_3 := D_2/\sqrt{2\pi}$.
\end{remark}

\subsection{Direct Part of $(\eps,\beta)$-Optimum Second-Order Coding Rates}

\begin{proof}[Proof of Proposition~\ref{lem:direct_second}]
We fix $n \in \mathbb{N}$. Also we let $\eps_n:=p_{\max}(\cC_n; W^n, P_{S^n})$, the maximum probability of error. Starting from the state-dependent Feinstein's lemma without cost constraints (Proposition~\ref{thm:feinstein-direct}), 
\begin{align}
\eps_n  \le\Pr\left[ \frac{1}{n} \log\frac{W^n(Y^n|X^n,S^n)}{ P_{X^n|S^n}W^n (Y^n|S^n)} \le\frac{1}{n}\log M + \gamma\right]+ \exp(-n\gamma)   \nonumber ,
\end{align}
where $\eps_n$ is the maximum error probability. 
We are going to let $\gamma :=\frac{\log n}{2n}$ throughout. Hence, the above bound becomes
\begin{align}
\eps_n \le\sum_{t\in\cP_n(\cS)}\Pr[T_{S^n}=t] \Pr\left[ \frac{1}{n} \log\frac{W^n(Y^n|X^n,S^n)}{ P_{X^n|S^n}W^n(Y^n|S^n) }\le\frac{1}{n}\log M + \gamma\,\Big|\, T_{S^n}=t\right]  + \frac{1}{\sqrt{n}}. \nonumber
\end{align}
Note that we have the freedom to choose the input distribution $P_{X^n|S^n}$ in the above expression. We are going to pick the most obvious choice $P_{X^n|S^n} = (P^*)^{\times n}$, i.e., 
$
P_{X^n|S^n}(x^n |s^n)=\prod_{k=1}^n  P^* (x_k|s_k).
$
Then,  by using the memorylessness of the channel $W^n$, the bound on the error becomes
\begin{equation}
\eps_n \le \sum_{t\in\cP_n(\cS)}\Pr[T_{S^n}=t] \underbrace{ \Pr\left[ \frac{1}{n}\sum_{k=1}^n\log \frac{W(Y_k|X_k, S_k)}{Q^*(Y_k|S_k)}\le R+\gamma\, \bigg|\, T_{S^n}=t\right] }_{=:\, \Xi^*(R; n, t)} + \frac{1}{\sqrt{n}},\nonumber
\end{equation}
where $Q^*(y|s):=\sum_x W(y|x,s) P^*(x|s)$. 

Now, we further condition on individual sequences within the type class $\cS_t := \{s^n\in\cS^n \,|\, T_{s^n} = t\}$.   We have,
\begin{align*}
\Xi^*(R; n, t) &= \sum_{s^n \in  \cS_t } \Pr[S^n =s^n|T_{S^n}=t] \Pr\left[ \frac{1}{n}\sum_{k=1}^n \log\frac{W(Y_k|X_k, S_k)}{Q^*(Y_k|S_k)}\le R+\gamma\, \bigg|\, S^n=s^n\right] \\
&= \Pr\left[ \frac{1}{n}\sum_{k=1}^n \log\frac{W(Y_k|X_k, S_k)}{Q^*(Y_k|S_k)}\le R+\gamma\, \bigg|\, S^n= (s^n)^* \right] ,
\end{align*}
where, since the inner probability is independent of $s^n \in \cS_t$, we chose an arbitrary representant $(s^n)^*\in \cS_t$. Furthermore, all the summands in the probability are independent random variables having the following statistics:
\begin{align}
\Exp\Bigg[\frac{1}{n}\sum_{k=1}^n\log \frac{W(Y_k|X_k, S_k)}{Q^*(Y_k|S_k)} \,\Bigg|\, S^n =(s^n)^* \Bigg] &=I(P^*, W|t)\nonumber\\
\Var\Bigg[\frac{1}{n}\sum_{k=1}^n \log\frac{W(Y_k|X_k, S_k)}{Q^*(Y_k|S_k)}\,\Bigg|\, S^n = (s^n)^*\Bigg]&= \frac{U(P^*, W|t)}{n} \nonumber ,
\end{align}
where $U(P^*,W|t)=\sum_s t(s) U(P^*_s , W_s)$.  Now,  by \cite[Lem.~62]{PPV10},  we know that the unconditional information variance coincides with the conditional information variance when evaluated at any capacity-achieving input distribution. Hence, 
\begin{equation}
U(P^*_s, W_s) = V(P^*_s, W_s),\qquad\forall\, s\in\cS,\nonumber
\end{equation}
since $P^*_s \in\cP(\cX)$ achieves capacity for channel $W_s\in\cP(\cY|\cX)$. Hence,  $ U(P^*, W|t)= V(P^*, W|t)$.  

As such, by the Berry-Esseen theorem~\cite[Sec.\ XVI.7]{feller}, 
\begin{equation}
\Xi^*(R; n, t)\le\Phi\left( \sqrt{n}\cdot \frac{R+\gamma-I(P^*, W|t) }{\sqrt{V(P^*, W|t)}} \right) + \frac{6\, L^+}{V_{\min}^{3/2}\, \sqrt{n}}. \nonumber
\end{equation}
Now using the fact that $\Phi(a+\eta)\le \Phi(a)+\eta/\sqrt{2\pi}$ for all $0 < a < 1- \eta \leq 1$ and the definition of  $\gamma=\frac{\log n}{2n}$, we have 
\begin{equation}
\Xi^*(R; n, t)\le\Phi\left( \sqrt{n}\cdot \frac{R-I(P^*, W|t) }{\sqrt{V(P^*, W|t)}} \right) + \frac{\log n}{2\sqrt{2\pi nV_{\min}}} + \frac{B}{\sqrt{n}}. \nonumber
\end{equation}
Combining all bounds, we have  \eqref{eqn:lemma_direct} as desired.  
\end{proof}

\subsection{Converse Part of $(\eps,\beta)$-Optimum Second-Order Coding Rates}\label{sec:prf_second_conv}

We will need the following  auxiliary lemmas. Recall the assumption that the channels $W_s$ have positive information dispersion, i.e.\ $V_s \geq V_{\min} > 0$.

\subsubsection*{Uniform Bound}

We will need this lemma which quantifies continuity properties for distributions around the unique capacity-achieving input distributions, $P_s^* \in \cP(\cX)$ for $s\in\cS$. We also write $P^*(x|s) = P_s^*(x)$. 

\begin{lemma}
  \label{lm:Pi-mu}
  Define $\Pi_{\mu}^s := \big\{ P \in \cP(\cX) \,\big|\, \twonorm{P - P^*(\cdot|s)} \leq \mu \big\}$ for small $\mu>0$.
  Then there exists a $\mu > 0$, as well as finite constants $\alpha > 0$ and  $\beta > 0$ such that the following holds.
  For all $s \in \cS$ and for all $P \in \Pi_{\mu}^s$, we have
  \begin{enumerate}
    \item[a)] $V(P, W_s) > \frac{V_{\min}}{2} >0$,
    \item[b)] $I(P, W_s) \leq C_s - \alpha  \twonorm{P - P^*(\cdot|s)}^2$,
    \item[c)] $ \big| V(P,W_s) - V(P^*(\cdot|s), W_s)\, \big| \leq \beta \twonorm{P - P^*(\cdot|s)}$. 
  \end{enumerate}
\end{lemma}

\begin{proof}
For a fixed $s\in\cS$, the Lemma holds due to the point-to-point channel coding results by Strassen~\cite{strassen} and Polyanskiy-Poor-Verd\'u~\cite[Lem.~49]{PPV10}. Also see   \cite[Lemma~7]{tomamicheltan12}. Now choose $\mu$ small enough (and $\alpha$ minimal, $\beta$ maximal) so that the claim holds for all $s\in\cS$.
\end{proof}

Recall now the definition of $\Xi_f(R; t)$ given in~\eqref{eqn:def_Xi}, which we repeat here as follows:
\begin{align*}
\Xi_{f}(R; n, t) = \Pr \Bigg[ \sum_{k=1}^n j_{T_{X^n|S^n}}(X_k;Y_k|S_k) \leq R \,\bigg|\, (X^n,S^n) = f(t) \Bigg] .
\end{align*}

We also define its generalized inverse $\Xi_f^{-1}(\eps; n, t) := \sup \{ R \,|\, \Xi_f(R; n, t) \leq \eps \}$ in the sense of Appendix~\ref{app:inverse}.

\begin{lemma}[Uniform Bound]
 \label{lm:psit}
  There exist finite constants $D  > 0$ and $N_0 \in \mathbb{N}$ such that the following statement holds. 
  For all $n \geq N_0$, $t \in \cP_n(\cS)$ and $\eps \in [0,1-(B+1)/\sqrt{n}]$, we have that
  \begin{align}
    \Xi_f^{-1}(\eps; n, t) \leq C(t) + \sqrt{\frac{V(t)}{n}} \Phi^{-1}\bigg(\eps + \frac{B}{\sqrt{n}} \bigg) + \frac{D \log n}{n}  \label{eq:toshow}
  \end{align}
   holds for all $f: \cT \to \cX^n \times \cS^n$ with $(x^n,s^n) = f(t)$ satisfying $s^n \in \cS_t$.
\end{lemma}

\begin{proof}
  Apply Lemma~\ref{lm:Pi-mu} to get $\mu$, $\alpha$ and $\beta$. Moreover, we define 
  $$C_s^{>\mu} := \sup_{P \in \cP(\cX) \setminus \Pi_{\mu}^s} I(P, W_s) \quad \textrm{and note that} \quad C_s^{>\mu} < C_s
  \quad \textrm{such that} \quad  \delta := \min_s \{C_s - C_s^{>\mu}\} > 0.$$
  We also need the uniform bounds $V(P, W) \leq V^+$ and $L(P,W) \leq L^+$.
 
  Let us fix $n$, $t$ and $\eps$ for the moment. We consider two cases for $(x^n,s^n) = f(t)$ as follows. Fix $0 < \xi < \frac{1}{|\cS|}$.
  \begin{itemize}
    \item[a)] For all $s \in \cS$, we have $T_{x^n|s^n}(\cdot|s) \in \Pi_{\mu}^s$ or $t(s) < \xi$.
    \item[b)] There exists an $s^* \in \cS$ with  $T_{x^n|s^n}(\cdot|s^*) \notin \Pi_{\mu}^s$ and $t(s^*) \geq \xi$.
  \end{itemize}
  
  \paragraph*{Case a)} We find that $V ( T_{x^n|s^n}, W | t ) \geq \frac{V_{\min}}{2} ( 1 - |\cS|\xi ) =: V_{-} > 0$ is uniformly bounded away from zero and thus do not hesitate to apply Berry-Esseen. This yields
  \begin{align}
    \Xi_f^{-1}(\eps; n, t) \leq I(T_{x^n|s^n}, W | t) + \sqrt{\frac{V(T_{x^n|s^n}, W | t)}{n}} \Phi^{-1} \bigg( \eps + \frac{B}{\sqrt{n}} \bigg) . \label{eq:berry}
  \end{align}
  Here, the constant $B = 6L^+/V_-^{3/2}$ can be chosen uniformly because $L(T_{x^n|s^n}, W|t) \leq L^+$. From the restriction $\eps \in [0,1-(B+1)/\sqrt{n}]$, $B \geq 1$ and $n \geq 4$, we get
  \begin{align*}
    \Bigg| \Phi^{-1}\bigg( \eps + \frac{B}{\sqrt{n}} \bigg) \Bigg| \leq -\Phi^{-1}\bigg(\frac1{\sqrt{n}}\bigg) \leq \sqrt{\log \frac{n}{4}} \leq \sqrt{\log n} \,,
  \end{align*}
  where we used the Chernoff bound $\Phi(x) \leq \frac{1}{2} \exp (-x^2/2)$ for $x < 0$.
   We partition $\cS = \cS_1 \cup \cS_2$ into two disjoint sets with $\cS_1$ containing all $s$ with $T_{x^n|s^n}(\cdot|s) \in \Pi_{\mu}^s$.
   We define $\zeta_s = \| T_{x^n|s^n}(\cdot|s) - 
  P^*(\cdot|s) \|$ for all $s \in \cS_1$. Then, from the continuity properties of Lemma~\ref{lm:Pi-mu}, we get
  \begin{align*}
    I(T_{x^n|s^n}, W | t) - C(t) &\leq - \alpha \sum_{s \in \cS_1} t(s) \zeta_s^2 - \sum_{s \in \cS_2} t(s) (C_s - C_s^{>\mu}) \leq - \alpha \sum_{s \in \cS_1} t(s) \zeta_s^2 - \sum_{s \in \cS_2} t(s) \delta \\
    \big| V(T_{x^n|s^n}, W | t) - V(t) \big| &\leq \beta \sum_{s \in \cS_1} t(s) \zeta_s + V^+ \sum_{s \in \cS_2} t(s)
  \end{align*}
  We may use that $x/(2\sqrt{a}) - x^2/(2\sqrt{a})^3 \leq \sqrt{a+x} - \sqrt{a} \leq x/(2\sqrt{a})$ for $a > 0$, $x > -a$ to bound
  $$\sqrt{a + x} \cdot  c \leq \sqrt{a} \cdot c + x \bigg( \frac{1}{2\sqrt{a}} + \frac{b}{(2\sqrt{a})^3} \bigg) \cdot |c| \qquad \forall a > 0, b > 0, c \in \mathbb{R}\quad \textrm{and}\quad -a < x < b .$$
  This is now applied to~\eqref{eq:berry}, to find
  \begin{align*}
     &\Xi_f^{-1}(\eps; n, t) - C(t) - \sqrt{\frac{V(t)}{n}} \Phi^{-1}\bigg( \eps + \frac{B}{\sqrt{n}} \bigg) \\
     &\quad \leq - \alpha \sum_{s \in \cS_1} t(s) \zeta_s^2 - \sum_{s \in \cS_2} t(s) \delta + \sqrt{\frac{\log n}{n}}\bigg( \frac{1}{2\sqrt{V_{\min}}} + \frac{V^+}{( 2 \sqrt{V_{\min}})^3} \bigg) \Bigg( \beta \sum_{s \in \cS_1} t(s) \zeta_s + V^+ \sum_{s \in \cS_2} t(s) \Bigg) \\
     &\quad = \sum_{s \in \cS_1} t(s) \Big( -\alpha \zeta_s^2 +  \gamma \beta \zeta_s \, \sqrt{\frac{\log n}{n}}  \Big) + \sum_{s \in \cS_2} t(s) \Big( - \delta +  \gamma V^+ \sqrt{\frac{\log n}{n}} \Big) ,
  \end{align*}
  where the shorthand $\gamma := 1/(2\sqrt{V_{\min}}) + V^+/(2\sqrt{V_{\min}})^3$ was introduced.
  
  Now, we first note that for the terms with $s \in \cS_2$ vanish asymptotically, since 
  \begin{align*}
    \gamma V^+ \sqrt{\frac{\log n}{n}} \leq \delta   
  \end{align*}
  for $n \geq N_0$ when $N_{0}$ is chosen appropriately large such that $\log N_0 \leq \delta^2/(\gamma V^+)^2 N_0$. 
    
  To analyze the terms with $s \in \cS_1$, we define $\zeta_s' = \sqrt{n}\zeta_s$ and write them as
  \begin{align*}
    \frac{1}{n} \big( -\alpha \zeta_s'^2 + \gamma\beta \zeta_s' \sqrt{\log n}\big) \leq \frac{\log n}{n} \bigg( \frac{\gamma \beta}{2 \sqrt{\alpha}} \bigg)^2,
  \end{align*}
  where the last inequality simply follows by maximizing the polynomial in $\xi_s'$. Thus, it suffices to choose $D := |\cS| (\gamma\beta)^2/(4\alpha)$, which implies~\eqref{eq:toshow}.
  
  \paragraph*{Case b)} Let $s^*$ be any element satisfying $T_{x^n|s^n}(\cdot|s) \notin \Pi_{\mu}^s$ and $t(s^*) > \xi$. Employing Chebyshev's inequality, we find
  \begin{align*}
    \Xi_f^{-1}(\eps; n, t) 
    &\leq I(T_{x^n|s^n}, W | t) + \sqrt{\frac{V(T_{x^n|s^n}, W | t)}{n(1-\eps)}} \\
    &\leq \sum_{s \neq s^*} t(s) C_s + t(s^*) I\big(T_{x^n|s^n}(\cdot|s^*), W(*|\cdot,s^*) \big) + \sqrt{\frac{V^+}{\sqrt{n}(B+1)}} \\
    &\leq C(t) - \delta \xi + \sqrt{\frac{V^+}{\sqrt{n}(B+1)}}\ .
  \end{align*}
  Thus, since $\delta > 0$ and $\xi > 0$, we find that~\eqref{eq:toshow} holds for all $n \geq N_0$, for $N_0$ chosen sufficiently large.
  
  Summarizing the analysis, we find that the proposition holds for $B$ and $D$ as defined above and any $N_0$ satisfying
  \begin{align*}
    N_0 \geq 4, \quad N_0 \geq \bigg( \frac{V^+}{\xi^2\delta^2(B+1)} \bigg)^2 \quad \textrm{and} \quad 
    \log N_0 \leq \frac{\delta^2}{(\gamma V^+)^2} N_0 \,.
  \end{align*}
 This concludes the proof.
\end{proof}

Finally, we proceed to prove Proposition~\ref{lem:converse_second}.

\begin{proof}[Proof of Proposition~\ref{lem:converse_second}]
  Fix $n \geq N_0$, where $N_0$ is taken from Lemma~\ref{lm:psit}.
  We recall Eq.~\eqref{eq:tobeusedagain}, which for every $\eps \in [0,1]$ gives the converse bound
\begin{align*}
    \frac{1}{n} \log M^*(\eps; W^n, P_{S^n}) \leq \sup_{f:\, \cT \to \cX^n \times \cS^n} \textrm{cv}(f; \eps, n) + \frac{\log n}{2n} + |\cX||\cS|\,\frac{\log (n+1)}{n} , 
\end{align*}
where we chose $\delta = 1/\sqrt{n}$ and
\begin{align*}
  \textrm{cv}(f; \eps, n) = \sup \Big\{  R \,\Big|\, \Exp \big[ \Xi_{f}(R; n, T_{S^n}) \big] \leq \eps + \delta \Big\}  .
\end{align*}

  We take the inverse of Lemma~\ref{lm:psit} using Lemma~\ref{lm:inverse} in Appendix~\ref{app:inverse}, which yields
  \begin{align*}
    \Xi_f(R; n, t) \geq
    \begin{cases}
       \Phi\Bigg( \frac{R - C(t) - D \frac{\log n}{n}}{\sqrt{ {V(t)}/{n}}} \Bigg) - \frac{B}{\sqrt{n}} & \textrm{if } R < R_0(t,n) \\
       1 - \frac{1}{\sqrt{n}} & \textrm{if } R \geq R_0(t,n)
    \end{cases} , 
  \end{align*}
 for all $n\ge N_0$,  where
    $R_0(t,n) := C(t) + \sqrt{\frac{V(t)}{n}} \Phi^{-1}\big(1 - \frac{1}{\sqrt{n}}\big) + \frac{D}{n}$. In particular, we may write
    \begin{align*}
      \Exp \big[\Xi_{f}(R; n, T_{S^n})  \big] &\geq \Exp \Bigg[ \Phi\Bigg( \sqrt{n} \cdot \frac{R - C(T_{S^n}) - D \frac{\log n}{n}}{\sqrt{V(T_{S^n})}} \Bigg) \Bigg] - \frac{B}{\sqrt{n}} - \frac{1}{\sqrt{n}} .
  \end{align*}
  Combining these results yields the claimed result.
\end{proof}

\subsection{Proof of Theorem~\ref{thm:sec_gen}}

\subsubsection{Direct part of Theorem~\ref{thm:sec_gen}}
Here we leverage on Proposition~\ref{lem:direct_second} to prove the direct part of  Theorem~\ref{thm:sec_gen}.
\begin{proof}
We set $r^*:=\sup\big\{r\, |\, K(r|\, \Ceps,\beta;W,P_{\hat{S}})\le\eps\big\}-\eta$ for some small $\eta>0$ and  we prove that $r^*$ is an $(\eps,\beta)$-achievable second-order coding rate. Clearly, 
\begin{equation}
K(r^* |\, \Ceps,\beta;W,P_{\hat{S}})\le\eps .  \label{eqn:K_smaller_eps} 
\end{equation}
For each $n \in \mathbb{N}$, set the number of codewords $M$ to be $\lfloor\exp( nC_\eps+ n^\beta r^* )\rfloor$. Trivially, 
\begin{align*}
\liminf_{n\to\infty}\frac{1}{n^{\beta}}\left[ \log M-n C_\eps \right]\ge r^*.
\end{align*}
Furthermore, Proposition~\ref{lem:direct_second} guarantees the existence of a length-$n$ code $\cC_n$ with $M$ codewords   satisfying
\begin{align}
p_{\max}(\cC_n; W^n, P_{S^n}) \le\Exp  \left[ \Phi \left( \sqrt{n}\cdot\frac{ C_\eps + n^{\beta-1} r^*  - C(T_{S^n})}{\sqrt{V( T_{S^n})}}\right) \right] +   \frac{D_1 \log n}{\sqrt{n}} + \frac{B+1}{\sqrt{n}} . \nonumber
\end{align}
Invoking the definition of $K(r  |\, \Ceps,\beta;W,P_{\hat{S}})$ in~\eqref{eqn:def_expect_T}  and taking the $\limsup$ on both sides yields
\begin{align*}
\limsup_{n\to\infty}\, p_{\max}(\cC_n; W^n, P_{S^n})\le K(r^* |\, \Ceps,\beta;W,P_{\hat{S}}) ,
\end{align*}
which together with \eqref{eqn:K_smaller_eps} yields $\limsup_{n\to\infty}p_{\max}(\cC_n; W^n, P_{S^n})\le\eps$. 
\end{proof}

\subsubsection{Converse part of Theorem~\ref{thm:sec_gen}}
Here we leverage on Proposition~\ref{lem:converse_second} to prove the converse part of Theorem~\ref{thm:sec_gen}. 
\begin{proof}
Let $\xi >0$  be any small positive constant.  We note from the definition of $(\eps,\beta)$-achievable second-order coding rate (cf.~Section~\ref{sec:defs_sec}) that there exists an integer $N_\xi$ such that any such  rate $r$ satisfies
\begin{equation*}
\frac{1}{n^\beta} \big[\log M^*(\eps; W, P_{S^n})-n C_\eps \big] \ge r-\xi
\end{equation*}
for all $n> N_\xi $. 
As a result, from Proposition~\ref{lem:converse_second}, 
\begin{align*}
r  -\xi &\le\liminf_{n\to\infty}\frac{1}{n^\beta} \big[\log M^*(\eps; W, P_{\hat{S}})-n C_\eps \big] \\
 &\le\liminf_{n\to\infty} \sup\Bigg\{ r \,\bigg|\, \Exp \Bigg[ \Phi\Bigg( \frac{n C_\eps + n^{\beta} r  - nC(T_{S^n})}{\sqrt{n V(T_{S^n})}} \Bigg) \Bigg] \leq \eps_n'\Bigg\}  ,
\end{align*}
where $\{\eps_n' \}_{n=1}^{\infty}$ is some sequence satisfying $\limsup_{n\to\infty}\eps_n'\le \eps$.   Then, by the definition of the $\liminf$ and the supremum, it is true that there exists another integer  $N_\xi'$ such that for all $n> \max\{ N_\xi , N_\xi'\}$, 
\begin{equation*}
\Exp \Bigg[ \Phi\Bigg( \frac{n C_\eps + n^{\beta}(r -2\xi)  - nC(T_{S^n})}{\sqrt{n V(T_{S^n})}} \Bigg) \Bigg] \leq \eps_n'.
\end{equation*}
Taking the $\limsup$ on both sides yields
\begin{equation}
K(r -2\xi|\, \Ceps,\beta; W, P_{\hat{S}})=\limsup_{n\to\infty}\Exp \Bigg[ \Phi\Bigg( \frac{n C_\eps + n^{\beta}(r -2\xi)  - nC(T_{S^n})}{\sqrt{n V(T_{S^n})}} \Bigg) \Bigg] \leq \eps . \nonumber
\end{equation}
The proof is then concluded by appealing to the same argument as that for the first-order case. See the argument succeeding \eqref{eq:implied}.
\end{proof}

\section{Proofs of the Specializations of   Theorem~\ref{thm:sec_gen} to   Various State Models}\label{sec:prf_examples}
In this section, we prove Theorems~\ref{thm:mixed}--\ref{thm:weird} using the general result in Theorem~\ref{thm:sec_gen}.
 
 \subsection{Mixed Channels}
\begin{proof}[Proof of Theorem~\ref{thm:mixed}]
In this case, for every $n$, the type $T_{S^n} \in\cP_n(\cS)$ can only take on two values, namely $T_{S^n}=(1,0)$ or $T_{S^n}=(0,1)$.  Set $\beta=\frac{1}{2}$. As such,
\begin{align}
  \Exp\left[ \Phi\left( \frac{n C_\eps + \sqrt{n} r -n C(T_{S^n})}{\sqrt{n V(T_{S^n})}}\right) \right]   =\alpha \Phi\left( \frac{nC_\eps + \sqrt{n} r -n C_{\mathrm{a}}}{\sqrt{n V_{\mathrm{a}}}}\right) + (1-\alpha) \Phi\left( \frac{nC_\eps + \sqrt{n} r -n C_{\mathrm{b}} }{\sqrt{n V_{\mathrm{b}}}}\right) \label{eqn:split_mixed}. 
\end{align}
In Case I, $C_\eps =C_{\mathrm{a}}=C_{\mathrm{b}}$ so the result follows directly. In Case II,  $C_\eps=C_{\mathrm{a}}<C_{\mathrm{b}}$ so the argument of the second term in \eqref{eqn:split_mixed} is negative and of the order $\Theta(\sqrt{n})$. This tends to $-\infty$ and so the second term vanishes as $n$ becomes large. Hence, the result. Case III follows similarly. 
\end{proof}

 \subsection{Independent and Identically Distributed (i.i.d.) States} \label{sec:prf_iid}
Before we begin the proof of Theorem~\ref{thm:iid}, we recall the setup and state and prove some auxiliary lemmas. Here, $S^n=(S_1, \ldots, S_n)$ is a sequence of i.i.d.\ random variables with common distribution $\pi\in\cP(\cS)$.  As usual, we assume that $V_{\min}:=\min_s V_s>0$ throughout. We first state and prove the following lemmas. 
\begin{lemma}[Approximating $V(T_{S^n})$ with $V(\pi)$] \label{lem:approx_V}
The following holds uniformly in $x\in\mathbb{R}$:
\begin{equation}
 \Exp \left[\Phi\left(\sqrt{n}\cdot\frac{x-C(T_{S^n})}{\sqrt{V(T_{S^n})}} \right)\right]-\Exp \left[\Phi\left(\sqrt{n}\cdot\frac{x-C(T_{S^n})}{\sqrt{V(\pi)}} \right)\right] = O\left(\frac{\log n}{n}\right).
\end{equation}
\end{lemma}
\begin{lemma}[Approximating $C(T_{S^n})$ with $C(\pi)$] \label{lem:approx_sum_Vs}
The following holds uniformly in $x\in\mathbb{R}$:
\begin{equation}
  \Exp \left[\Phi\left(\sqrt{n}\cdot\frac{x-C(T_{S^n})}{\sqrt{V(\pi)}} \right)\right] - \Phi\left(\sqrt{n}\cdot\frac{x-C(\pi)}{\sqrt{V(\pi)+V^*(\pi)}} \right)  = O\left( \frac{1}{\sqrt{n}} \right),
\end{equation}
where $V^*(\pi)$ is defined in \eqref{eqn:V_star}. 
\end{lemma}

\begin{proof}[Proof of Lemma~\ref{lem:approx_V}]
Fix $n \in \mathbb{N}$. For the sake of brevity, let 
\begin{equation}
\Phi_1 := \Phi\left(\sqrt{n}\cdot\frac{x-C(T_{S^n})}{\sqrt{V(T_{S^n})}} \right),\qquad \Phi_2 := \Phi\left(\sqrt{n}\cdot\frac{x-C(T_{S^n})}{\sqrt{V(\pi)}}\right). \nonumber
\end{equation}
Also consider a {\em typical set} defined as follows:
\begin{align}
\cA_n(\pi):=\left\{s^n\in\cS^n \, \bigg| \,  \| T_{s^n}-\pi\|_{\infty}<\sqrt{ \frac{\log n}{n}}\right\} \label{eqn:typicalset}
\end{align}
From Lemma~\ref{lem:conc1} in Appendix~\ref{app:concentations}, we know that if $S^n$ is generated independently from $\pi$, 
\begin{align*}
\Pr[ S^n\notin\cA_n (\pi)]\le\frac{2|\cS|}{n^2}.
\end{align*}
We are interested in the function $h(v):=\Phi\big(\frac{\alpha}{v}\big)$. Let $\phi(x)=\frac{1}{\sqrt{2\pi}}e^{-x^2/2}$ be the Gaussian PDF and $\phi^{(k)}$ its $k$-th derivative.  By direct differentiation, the first few derivatives of $h$ are
\begin{align*}
h'(v) &=\frac{\mathrm{d}}{\mathrm{d}v}\Phi\Big(\frac{\alpha}{v}\Big)=-\frac{\alpha}{2v^{3/2}}\phi\Big(\frac{\alpha}{v}\Big)=\frac{1}{2v}\phi'\Big(\frac{\alpha}{v}\Big),\\
h''(v) & = -\frac{1}{2v^2}\phi'\Big(\frac{\alpha}{v}\Big) - \frac{\alpha}{4v^{5/2}}\phi''\Big(\frac{\alpha}{v}\Big)=\frac{1}{4v^2}\phi^{(3)}\Big(\frac{\alpha}{v}\Big) , \\
h^{(k)}(v)  &= \frac{1}{(2v)^k }\phi^{( 2k-1)}\Big(\frac{\alpha}{v}\Big),
\end{align*}
where we used the facts that $\phi'(x) = -x\phi(x)$ and $\phi^{(k+1)} (x)=-x\phi^{(k)}(x) - k \phi^{(k-1)}(x)$ for $k\ge 1$.  Hence, by  setting $\alpha:= \sqrt{n}\cdot  [x-C(T_{S^n})]$ and Taylor expanding $v\mapsto h(v)$  around $v=V(\pi)$,  we obtain the expansion
\begin{align*}
\Phi_1 = \Phi_2 + \sum_{k=1}^{\infty}\frac{1}{k!}\left(\frac{V(T_{S^n})-V(\pi)}{2V(\pi)}\right)^k \phi^{(2k-1)}\left(\frac{x-C(\pi)}{\sqrt{V(\pi)}} \right).
\end{align*}
Taking expectations on both sides, we obtain
\begin{align*}
\Exp\big[\Phi_1 - \Phi_2\big]= \Exp\left[\, \sum_{k=2}^{\infty}\frac{1}{k!}\left(\frac{V(T_{S^n})-V(\pi)}{2V(\pi)}\right)^k \phi^{(2k-1)}\left(\frac{x-C(\pi)}{\sqrt{V(\pi)}} \right)\right]
\end{align*}
where the $k=1$ term vanished because $V(\cdot)$ is linear and the expectation of   $T_{S^n}(s)$ is exactly $\pi(s)$ for every $s\in\cS$.  Now let the random variable in the expectation on the RHS be denoted as $\Delta_n$, i.e., 
\begin{align*}
\Delta_n:=\Phi_1 - \Phi_2-\left(\frac{V(T_{S^n})-V(\pi)}{2V(\pi)}\right)\phi'\left(\frac{x-C(\pi)}{\sqrt{V(\pi)}} \right).
\end{align*}
Observe that with probability one, 
\begin{align*}
\big|   \Delta_n \big|  \le  2+ \frac{1}{8}\left(\frac{V^+}{V_{\min}  }+1 \right) 
\end{align*}
because $\sup_{x\in\mathbb{R}}|\phi'(x)|\le \frac{1}{4}$ and the information dispersions are bounded above and below by $V^+$ and $V_{\min}$ respectively.   By the law of total expectation, 
\begin{align*}
\Exp\big[\Delta_n\big]= \Exp\big[\Delta_n\,\big|\, S^n\in\cA_n\big]\Pr\big[ S^n\in\cA_n \big] + \Exp\big[\Delta_n\,\big|\, S^n\notin\cA_n\big]\Pr\big[ S^n\notin\cA_n \big]
\end{align*}
Hence, by the triangle inequality,
\begin{align}
\left|\, \Exp\big[\Delta_n\big] \, \right|   \le \left|\, \Exp\big[\Delta_n\,\big|\, S^n\in\cA_n\big] \, \right|  + \left[2+ \frac{1}{8}\left(\frac{V^+}{V_{\min}  }+1 \right) \right]\cdot \frac{2|\cS|}{n^2}. \label{eqn:total_prob} 
\end{align}
 Now we bound the expectation on the RHS above as follows:
\begin{align}
\big| \Exp[\Delta_n\,\big|\,  S^n\in\cA_n (\pi) ] \big|\le\sum_{k=2}^{\infty}\frac{1}{k!} \left| \Exp \left[ \left(\frac{V(T_{S^n})-V(\pi)}{2V(\pi)}\right)^k   \, \bigg|  \, S^n\in \cA_n(\pi) \right] \right|  \cdot \left| \phi^{(2k-1)}\left(\frac{x-C(\pi)}{\sqrt{V(\pi)}} \right)\right|. \label{eqn:infinite_sum}
\end{align}
Now, we have
\begin{align}
\left| \Exp \left[ \left(\frac{V(T_{S^n})-V(\pi)}{2V(\pi)}\right)^k   \, \bigg|  \, S^n\in \cA_n(\pi) \right] \right| \le \Exp \left[ \left|\frac{V(T_{S^n})-V(\pi)}{2V(\pi)}\right|^k   \, \bigg|  \, S^n\in \cA_n(\pi) \right]   \le  \left(\frac{|\cS| V^+}{2V_{\min}} \sqrt{ \frac{\log n}{n}} \right)^k \nonumber
\end{align}
because  $V(\pi)\ge V_{\min}$ and 
\begin{align*}
\big|V(T_{s^n})-V(\pi)  \big| = \left| \sum_s \big[T_{s^n}(s)-\pi(s)\big]  \, V_s  \right|\le |\cS| V^+ \sqrt{ \frac{\log n}{n}}
\end{align*}
for all sequences $s^n\in\cA_n (\pi)$.  Furthermore, by Lemma~\ref{lem:gaussian_der} in Appendix~\ref{app:gaussian_der}, we know that
\begin{align*}
\sup_{x\in \mathbb{R}} \big|\phi^{(k)}(x)\big|\le  \frac{e^{1/8}}{(2\pi)^{1/4}} \,  k^{1/4} \, \left(\frac{k}{e}\right)^{k/2}
\end{align*}
for all $k \in \mathbb{N}$.  Note that this upper bound is monotonically increasing in $k$ so we can upper bound $\sup_{x\in \mathbb{R}}  |\phi^{(2k-1)}(x) |$ with the upper bound  for  $\sup_{x\in \mathbb{R}}  |\phi^{(2k)}(x) |$ which is what we do in the following. Let $c=e^{1/8}  (2\pi)^{1/4}$ be a universal constant.  As a result,  the infinite sum in \eqref{eqn:infinite_sum} can be bounded as 
\begin{align}
\big| \Exp[ \Delta_n  \, \big|\,  S^n\!\in\!\cA_n (\pi) ] \big| \!\le\! c \sum_{k=2}^{\infty} \frac{1}{{k}^{1/4}\left( \frac{k}{e}\right)^k }    \left(\frac{|\cS| V^+}{2V_{\min}} \sqrt{ \frac{\log n}{n}} \right)^k \left(\frac{2k}{e}\right)^{k}  \!\le\! c  \sum_{k=2}^{\infty}\left(\frac{|\cS| V^+}{  V_{\min}} \sqrt{ \frac{\log n}{n}} \right)^k \!   = \! O\left( \frac{\log n}{n}\right) \label{eqn:infinite_series}
\end{align}
where in the first step, we used Stirling's bounds for the factorial \cite[Statement 6.1.38]{Abr65}, and in the last step, we used the formula for the infinite sum of geometric series.  Note that the implied constant in the $O(\cdot)$-notation depends only on $|\cS|, V_{\min}$ and $V^+$.  Uniting \eqref{eqn:total_prob} and \eqref{eqn:infinite_series} completes the proof.
\end{proof}

\begin{proof}[Proof of Lemma~\ref{lem:approx_sum_Vs}]
Let us introduce the following random variables
\begin{align}
J_n:=\frac{1}{\sqrt{n}}\sum_{k=1}^n E_k,\qquad\mbox{and}\qquad E_k:=\frac{1}{\sqrt{V^*(\pi)}}\big[ C_{S_k}-C(\pi)\big]. \label{eqn:defJ}
\end{align}
Note that $J_n$ converges in distribution to a standard Gaussian since $E_k$ are i.i.d.\ zero-mean, unit-variance  random variables. Furthermore, the third absolute moment of each $E_k$ is bounded above by $2  \log^3 |\cY| /\sqrt{V^*(\pi)^3}$.  Since $C(T_{S^n})=\frac{1}{n}\sum_{k=1}^n C_{S_k}$, we can write
\begin{align*}
 \Exp \left[\Phi\left(\sqrt{n}\cdot\frac{x-C(T_{S^n})}{\sqrt{V(\pi)}} \right)\right]= \Exp \left[\Phi\left(\sqrt{n}\cdot\frac{x-C(\pi)}{\sqrt{V(\pi)}} +\sqrt{\frac{V^*(\pi)}{V(\pi)}} \, J_n  \right)\right]
\end{align*}
Now by using the weak form of the Berry-Esseen theorem \cite[Thm.~2.2.14]{tao12}, we claim that the following holds:
\begin{align}
\Exp \left[\Phi\left(\sqrt{n}\cdot\frac{x-C(T_{S^n})}{\sqrt{V(\pi)}} \right)\right]=\Exp \left[\Phi\left(\sqrt{n}\cdot\frac{x-C(\pi)}{\sqrt{V(\pi)}} +\sqrt{\frac{V^*(\pi)}{V(\pi)}}  \,  Z  \right)\right] + O\left( \frac{1}{\sqrt{n}}\right) , \label{eqn:weakBE}
\end{align}
where $Z$ is a standard Gaussian random variable.    Note that the preceding estimate is uniform in $x \in \mathbb{R}$.  Let us justify~\eqref{eqn:weakBE} more carefully using the argument suggested in~\cite[Sec.\ 2.2.5]{tao12}. Define the family of functions $g_n : \mathbb{R}\cup\{\pm\infty\} \to [0,1]$ as    
\begin{align*}
g_n(\theta):=  \Phi\left(\sqrt{n}\cdot\frac{x-C(\pi)}{\sqrt{V(\pi)}} +\sqrt{\frac{V^*(\pi)}{V(\pi)}} \,  \theta  \right) .
\end{align*}
Observe that $\int_{-\infty}^{\infty} g_n'(\theta)\, \mathrm{d}\theta=1$.  Now, let $f_{{J}_n}$ and $\phi$ be the PDFs corresponding to  $J_n$ and $Z$ respectively. Then,
\begin{align*}
\big| \Exp[ g_n(J_n)] - \Exp[  g_n (Z) ] \big| &= \left| \int_{-\infty}^{\infty} g_n(\theta) f_{{J}_n}(\theta)   \, \mathrm{d}\theta   -\int_{-\infty}^{\infty} g_n(\theta) \phi(\theta)   \, \mathrm{d}\theta      \right|\\
&=\left|g_n(\infty) -  \int_{-\infty}^{\infty} \Pr[ J_n\le\theta]g_n'(\theta)\,\mathrm{d}\theta - g_n(\infty) +  \int_{-\infty}^{\infty} \Pr[ Z\le\theta]g_n'(\theta)\,\mathrm{d}\theta \right|  \\
&\le\sup_{\theta\in \mathbb{R}} \big |\Pr[ J_n<\theta]- \Pr[ Z<\theta] \big|   \int_{-\infty}^{\infty}  g_n'(\theta)\,\mathrm{d}\theta    \le \frac{12\, \log^3 |\cY|  }{\sqrt{ n V^*(\pi)^3 }} ,
\end{align*}
where the first inequality follows by integration by parts and the final inequality by the usual Berry-Esseen theorem~\cite[Sec.\ XVI.7]{feller}. This proves~\eqref{eqn:weakBE}.
Now, we introduce $\phi(x; \mu,\sigma^2)$ as the Gaussian PDF with mean $\mu$ and variance $\sigma^2$ (so $\phi(x) = \phi(x;0,1)$) and write
     \begin{align*}
       \Exp \Bigg[ \Phi \bigg( \sqrt{n} \cdot \frac{x - C(\pi)}{\sqrt{V(\pi)}}+  \sqrt{\frac{V^*(\pi)}{V(\pi)}}  \, Z \bigg) \Bigg] 
       &= \int_{-\infty}^{\infty} \phi(z)\, \Phi \bigg( \sqrt{n} \cdot \frac{x - C(\pi)}{\sqrt{V(\pi)}}+  \sqrt{\frac{V^*(\pi)}{V(\pi)}}  \, z \bigg)\; \mathrm{d}  z   \\
       &= \int_{-\infty}^{\infty} \phi \bigg(z; 0, \frac{V^*(\pi)}{n} \bigg) \int_{-\infty}^{x+z} \phi \bigg( w ; C(\pi), \frac{V(\pi)}{n} \bigg) \; \mathrm{d}  w\;   \mathrm{d}  z \\
       &= \int_{-\infty}^{x} \int_{-\infty}^{\infty}  \phi \bigg(z; 0, \frac{V^*(\pi)}{n} \bigg)  \phi \bigg(w - z; C(\pi), \frac{V(\pi)}{n} \bigg)  \; \mathrm{d}  z\; \mathrm{d}  w  .
     \end{align*}
     Now the inner integral is obviously a convolution integral, under which mean and variance of independent Gaussians are additive. Thus,
     \begin{align*}
       \int_{-\infty}^{x}  \int_{-\infty}^{\infty} \phi \bigg(z; 0, \frac{V^*(\pi)}{n} \bigg)  \phi \bigg( w - z; C(\pi), \frac{V(\pi)}{n} \bigg)  \; \mathrm{d}  z\;  \mathrm{d}  w 
       &= \int_{-\infty}^{x} \phi \bigg( w; C(\pi), \frac{V(\pi)+V^*(\pi)}{n} \bigg) \; \mathrm{d}   w \\
       &= \Phi \bigg( \sqrt{n} \cdot \frac{x - C(\pi)}{\sqrt{V(\pi) + V^*(\pi)}} \bigg).
     \end{align*}
This, together with the estimate in~\eqref{eqn:weakBE}, proves Lemma~\ref{lem:approx_sum_Vs}.
\end{proof}

\begin{proof}[Proof of Theorem~\ref{thm:iid}]
Combining Lemmas~\ref{lem:approx_V} and \ref{lem:approx_sum_Vs} yields 
\begin{align*}
 \Exp \left[\Phi\left(\sqrt{n}\cdot\frac{x-C(T_{S^n})}{\sqrt{V(T_{S^n})}} \right)\right]  = \Phi\left( \sqrt{n} \cdot \frac{x-C(\pi)}{\sqrt{V(\pi) + V^*(\pi)}} \right)+ O\left(\frac{1}{\sqrt{n}} \right)
\end{align*}
uniformly in $x\in\mathbb{R}$. Since $C(\pi)$ is the $\eps$-capacity in this i.i.d.\ state scenario, 
\begin{align*}
K\bigg(r \, \Big|\,C_\eps , \frac{1}{2};W,P_{\hat{S}}\bigg)  &= \limsup_{n\to\infty}  \Exp \left[\Phi\left( \frac{n  C_\eps +\sqrt{n}r-nC(T_{S^n})}{\sqrt{nV(T)}} \right)\right] \\
&= \limsup_{n\to\infty}  \Phi\left(  \sqrt{n}\cdot \frac{  C_\eps +r/\sqrt{n}- C(\pi)}{\sqrt{V(\pi)+ V^*(\pi)}} \right) =  \Phi\left(   \frac{ r}{\sqrt{V(\pi)+ V^*(\pi)}} \right) .
\end{align*}
Consequently, an application of Theorem~\ref{thm:sec_gen} yields the desired result.
\end{proof}

\subsection{Block i.i.d.\ States}
Here, we prove Theorem~\ref{thm:subblock} concerning the second-order asymptotics for the channel with block i.i.d.\ states. Recall that $\nu\in(0,1)$ and  for each blocklength $n$, there are $d+1$ subblocks, the first $d=\lfloor n^\nu\rfloor$ of which have length $m$ and the remaining subblock has length $r = n - md$.  Each block is assigned an independent state from $\pi \in\cP(\cS)$. The proof here only requires a slight modification of the proof for the i.i.d.\ case. See Section~\ref{sec:prf_iid} for the overall outline. Essentially, we need to provide analogues of Lemma~\ref{lem:approx_V} and~\ref{lem:approx_sum_Vs}.  

\begin{proof}[Proof of Theorem~\ref{thm:subblock}]
First, since the residual $O(\frac{\log n}{n})$  term in Lemma~\ref{lem:approx_V} depends on how rapidly the type $T_{S^n}$ concentrates to $\pi$, in this block i.i.d.\  setting,  it is easy to see that since we have $d=\lfloor n^\nu \rfloor$ independent random variables drawn from $\pi$, 
\begin{equation}
  \Exp \left[\Phi\left(\sqrt{n}\cdot\frac{x-C(T_{S^n})}{\sqrt{V(T_{S^n})}} \right)\right]-\Exp \left[\Phi\left(\sqrt{n}\cdot\frac{x-C(T_{S^n})}{\sqrt{V(\pi)}} \right)\right]   = O\left(\frac{\log n}{n^\nu}\right), \nonumber
\end{equation}
uniformly in  $x\in\mathbb{R}$. 

Second, to develop an analogue of Lemma~\ref{lem:approx_sum_Vs}, we define, instead of $J_n$ and $E_k$ in \eqref{eqn:defJ}, the following random variables
\begin{align*}
\tilde{J}_{d+1}:=\frac{1}{m\sqrt{d+1}}\sum_{j=1}^{d+1} \tilde{E}_j,\qquad  \tilde{E}_j&:=\frac{m}{\sqrt{V^*(\pi)}}\big[ C_{S_{ (j-1) m+1}}-C(\pi)\big], \quad 1\le j\le d, \\
 \mbox{and}\qquad \tilde{E}_{d+1}&:=\frac{r}{\sqrt{V^*(\pi)}}\big[ C_{S_{dm+1}}-C(\pi)\big] .  
\end{align*}
The  random variables  $\tilde{E}_j$ for $1\le j \le d+1$  are independent and are basically the sums of the $E_k$'s  within each subblock. Note that $\Var[ \tilde{J}_{d+1}] = \frac{dm^2+r^2}{m^2(d+1)}$ which converges to one as $n$ becomes large. Thus, $\tilde{J}_{d+1}$ converges in distribution to a standard Gaussian. Define 
$$
\tau_n :=  \frac{m^2 (d+1)}{n} \sim n^{1-\nu}  ,
$$  
where $a_n \sim b_n$ means that the limit of ratio of $a_n$ and $b_n$ tends to one. 
Following the proof of Lemma~\ref{lem:approx_sum_Vs}, we obtain  
\begin{equation}
 \Exp \left[\Phi\left(\sqrt{n}\cdot\frac{x-C(T_{S^n})}{\sqrt{V(\pi)}} \right)\right] - \Phi\left(\sqrt{n}\cdot\frac{x-C(\pi)}{\sqrt{V(\pi)+  \tau_n  V^*(\pi)}} \right)  = O\left( \frac{1}{n^{\nu/2}} \right), \nonumber
\end{equation}
which holds  uniformly in  $x\in\mathbb{R}$.  As a result, 
\begin{align*}
 \Exp \left[\Phi\left(\sqrt{n}\cdot\frac{x-C(T_{S^n})}{\sqrt{V(T_{S^n})}} \right)\right]  = \Phi\left( \sqrt{n} \cdot \frac{x-C(\pi)}{\sqrt{V(\pi) + \tau_n   V^*(\pi)}} \right)+  O\left( \frac{1}{n^{\nu/2}} \right),
\end{align*}
also holds uniformly in  $x\in\mathbb{R}$.
Now  since   $\tau_n    V^*(\pi)$ dominates $V(\pi)$, we see that 
\begin{align*}
K\bigg(r \, \Big|\, C_\eps , 1-\frac{\nu}{2};W,P_{\hat{S}}\bigg) & =\limsup_{n\to\infty}  \Exp \left[\Phi\left( \frac{nC_\eps  +n^{1-\nu/2} r-nC(T_{S^n})}{\sqrt{nV(T_{S^n})}} \right)\right]   \\
& = \limsup_{n\to\infty}  \Phi\left(  \sqrt{n}\cdot \frac{ C_\eps  +n^{-\nu/2}r- C(\pi)}{\sqrt{V(\pi)+ \tau_n  V^*(\pi)}} \right) =  \Phi\left(   \frac{ r}{\sqrt{ V^*(\pi)}} \right) .
\end{align*}
As such, an application of  Theorem \ref{thm:sec_gen} reveals that the $(\eps, 1-\frac{\nu}{2})$-optimum second-order coding rate is   $\sqrt{ V^*(\pi)}\Phi^{-1}(\eps)$ as desired.
\end{proof}

\subsection{Markov States}
In this Section, we prove Theorem~\ref{thm:markov} concerning the second-order asymptotics for the channel with Markov states. Recall that the    Markov chain is time-homogeneous, irreducible and ergodic, which means there exists a unique stationary distribution $\pi \in \cP(\cS)$.  As in the block i.i.d.\ case, the proof is very similar to that for the i.i.d.\ case  in Section~\ref{sec:prf_iid}. As such, we only highlight the salient differences.

First, we modify  the definition of the typical set $\cA_n(\pi)$ in  \eqref{eqn:typicalset} to the following:
\begin{align}
\cA_n(\pi;\zeta):=\left\{s^n\in\cS^n \, \bigg| \, \|T_{s^n}-\pi\|_{\infty} <  \zeta\, \sqrt{\frac{\log n}{n}}\right\} \label{eqn:typ_set2}
\end{align}
For the Markov case, Lemma~\ref{lem:markov_conc} in Appendix~\ref{app:concentations} shows that 
\begin{align*}
\Pr[ S^n\notin\cA_n (\pi;\zeta)] = O\left(\frac{1}{n}\right).
\end{align*}
if $\zeta>0$ is chosen sufficiently large (depending on the mixing  properties of the chain).  This allows us to prove the same order estimate as the one in Lemma~\ref{lem:approx_V}.

Second, we note that since the chain is time-homogeneous, irreducible and ergodic, it is $\alpha$-mixing and in fact, the $\alpha$-coefficient tends to zero exponentially fast~\cite[Thm.~3.1]{Bradley}. As such, the following surrogate of the Berry-Esseen theorem due to Tikhomirov~\cite{Tik} can be used.  
\begin{theorem} \label{thm:modified_BE}
Suppose that a stationary zero-mean process $\{X_j\}_{j\ge 1}$ is $\alpha$-mixing, where 
\begin{align*}
\alpha(n)\le Ke^{-\kappa n}, \qquad \Exp\big[ |X_1|^{4+\gamma}\big]<\infty,\qquad \sigma_n^2 \to \infty,
\end{align*} 
where 
\begin{align*}
\sigma_n^2 := \Exp\bigg[ \Big(\sum_{j = 1}^n X_j \Big)^2\bigg].
\end{align*}
Then,  there is a constant $B >0$, depending on $K,\kappa$ and $\gamma$ such that for all $n\ge 1$,  
\begin{align*}
\sup_{x\in\mathbb{R}} \bigg| \Pr \bigg[ \frac{1}{\sigma_n}  \sum_{j=1}^n X_j \le x  \bigg] -\Phi(x) \bigg|\le\frac{B\log n}{\sqrt{n}}.
\end{align*}
\end{theorem} 

\begin{proof}[Proof of Theorem~\ref{thm:markov}]
The estimate provided in Lemma~\ref{lem:approx_V} remains unchanged if we replace the definition  $\cA_n(\pi)$ in  \eqref{eqn:typicalset}  with $\cA_n(\pi;\zeta)$ in \eqref{eqn:typ_set2} for a suitably large constant $\zeta>0$. By replacing the use of the usual Berry-Esseen theorem in  the proof of Lemma~\ref{lem:approx_sum_Vs} with  Theorem~\ref{thm:modified_BE}, we see that the following estimate holds uniformly in $x\in\mathbb{R}$
\begin{align*}
\Exp \left[\Phi\left(\sqrt{n}\cdot\frac{x-C(T_{S^n})}{\sqrt{V(\pi)}} \right)\right] - \Phi\left(\sqrt{n}\cdot\frac{x-C(\pi)}{\sqrt{V(\pi)+V^{**}_n(M)}} \right)  = O\left( \frac{\log n}{\sqrt{n}} \right),
\end{align*}
where  the effective variance for blocklength $n$ is
\begin{align*}
V^{**}_n (M) : = \frac{1}{n}\Var\Bigg[ \sum_{k=1}^n  C_{S_k}  \Bigg] = \frac{1}{n} \sum_{ k,l=1}^n  \Cov\big[ C_{S_k}, C_{S_l} \big] = \Var[C_S] + \frac{2}{n}\sum_{j=1}^n (n-j) \Cov\big[ C_{S_j}, C_{S_{1+j}} \big].
\end{align*}
Since $|\!\Cov [ C_{S_1}, C_{S_{1+j}} ]|$ decays exponentially fast in the lag $j$ for these chains, $ \frac{1}{n}\sum_{j=1}^{\infty}  j\Cov [ C_{S_1}, C_{S_{1+j}} ] \to 0$. Hence, the result for this Markov state scenario follows from an application of Theorem~\ref{thm:sec_gen}.
\end{proof}

\subsection{Memoryless but Non-Stationary States}

\begin{proof}[Proof of Theorem~\ref{thm:weird}]
In this case, for every $n\in\mathbb{N}$, the expectation that defines $K(r |\, \Ceps, \frac{1}{2}; W, P_{\hat{S}})$ in \eqref{eqn:def_expect_T} simplifies. In particular,
\begin{align}
\Exp\left[ \Phi\left( \sqrt{n}\cdot \frac{ C_\eps  + r/\sqrt{n}-C(T_{S^n})}{\sqrt{nV(T_{S^n})}}\right)\right]=\Phi\left( \sqrt{n}\cdot \frac{ C_\eps   + r/\sqrt{n}-C(t_n)}{\sqrt{nV(t_n)}}\right) \label{eqn:phi_weird}
\end{align}
for some specific type $t_n\in\cP_n(\cS)$ for each $n$.  We need to evaluate the $\limsup$ of the above. Recall that one definition of the  $\limsup$ is that it is the supremum of all   subsequential limits.   Now consider the subsequence in $\mathcal{J}$ indexed by $n_k:=2^{2k}-1$ for $k\in\mathbb{N}$. Along this subsequence, the type of $S^{n_k}$ is $t_{n_k}=(\frac{2}{3},\frac{1}{3})$ for all $k$ and so $C_\eps  =C(t_{n_k})=\frac{2C_0}{3}+\frac{C_1}{3}$. The  $\limsup$ is achieved along this $\{n_k\}_{k\ge 1}$. All other subsequences   $\{n_{l}\}_{l\ge 1}$ that have a different subsequential limit  of $\Phi$ on the RHS of~\eqref{eqn:phi_weird}   result in a strictly smaller   $\lim_{l\to\infty}C(t_{n_{l}})$ (because $C_0<C_1$)  so  the  $\limsup$    vanishes.
\end{proof}

\appendix

\subsection{Evaluating the Covariance Term for Markov States}
\label{app:markov}

Here, we assume that $\hat{S} = \{ S^n = (S_1, \ldots, S_n)\}_{n=1}^{\infty}$ evolves according to a time-homogenous, irreducible, ergodic Markov chain, specializing Example~\ref{eg:markov}. Let $M$ be transition matrix of the Markov chain governing the memory, i.e.\ the probability for moving from state $s \in \cS$ to state $s' \in \cS$ in $k$ steps is given by $\Pr[S_{\ell+k} = s' | S_{\ell} = s] = [M^k]_{s,s'}$. We assume for the following that $M$ is diagonalizable.\footnote{Indeed, the following arguments can be generalized by using the Jordan normal form instead of the eigenvalue decomposition in case $M$ is not diagonalizable.} 
Thus, $M = U \Sigma U^{\dagger}$ with diagonal matrix $\Sigma = \diag(\lambda_1, \lambda_2, \ldots \lambda_{|\cS|})$ for eigenvalues satisfying
 $1 = |\lambda_1| > |\lambda_2| \geq |\lambda_3| \geq \ldots \geq |\lambda_{|\cS|}|$. 

Recall that the dispersion is composed of $V(\pi)$ and the term
$$V^{**}(M) = \Var_{S \leftarrow \pi} [C_S] + 2 \sum_{k=1}^{\infty} \Cov_{S_1 \leftarrow \pi} [C_{S_1}, C_{S_{1+k}}] .$$
Here, $S_1 \leftarrow \pi$ where $\pi$ is the stationary distribution of $M$ and the joint distributions are induced the Markov chain $S_1 \rightarrow S_2 \rightarrow \ldots \rightarrow S_n$ with (invariant) transition kernel $M$.

\begin{lemma}
  \label{lm:markov-simplified}
   Let $\hat{S}$ be governed by a diagonalizable transition matrix $M = U \diag ( 1, \lambda_2, \ldots, \lambda_{|\cS|}) U^{\dagger}$ of a time-homogenous, irreducible, ergodic Markov chain with stationary distribution $\pi$. Then,
   \begin{align}
    &V^{**}(M) = \Cov_{(S, S') \leftarrow \pi \times \Pi} \big[ C_S, C_{S'} \big]
    \quad \textrm{where} \qquad \Pi(s'|s) = \left[ U \diag\left( 1, \frac{1+\lambda_2}{1-\lambda_2},  \ldots, \frac{1+ \lambda_{|\cS|}}{1-\lambda_{|\cS|}} \right) U^{\dagger} \right]_{s,s'}. \label{eq:vdoublestar}
   \end{align}
\end{lemma}

\begin{remark}
  This is particularly simple to evaluate for the Gilbert-Elliott-type state evolution with $\cS = \{0, 1\}$ and transition probability $0 < \tau < 1$. This can be described by the transition matrix
  \begin{align*}
   M = \left( \begin{array}{cc} 1-\tau & \tau \\ \tau & 1-\tau \end{array} \right) = \frac{1}{2} \left( \begin{array}{cc} 1 & 1 \\ 1 & -1 \end{array} \right)\left( \begin{array}{cc} 1 & 0 \\ 0 & 1-2\tau \end{array} \right)\left( \begin{array}{cc} 1 & 1 \\ 1 & -1 \end{array} \right) , 
  \end{align*}
  with eigenvalues $1$ and $\lambda_2 = 1 - 2\tau$, where $|\lambda_2| < 1$. It is easy to verify that
  \begin{align*}
    \Pi = \frac{1}{2} \left( \begin{array}{cc} 1 & 1 \\ 1 & -1 \end{array} \right)\left( \begin{array}{cc} 1 & 0 \\ 0 & \frac{1-\tau}{\tau} \end{array} \right)\left( \begin{array}{cc} 1 & 1 \\ 1 & -1 \end{array} \right) = \frac{1}{2\tau} \left( \begin{array}{cc} 1 & 2\tau -1 \\ 2\tau - 1 & 1 \end{array} \right)
  \end{align*}
  and the covariance term can then be evaluated to $V^{**}(M) = \frac{1-\tau}{4\tau} (C_0 - C_1)^2$, generalizing the dispersion found for the related problem in~\cite[Thm.~4]{PPV11}.
\end{remark}

\begin{proof}
  First, let us verify that $\Pi(s'|s)$ is indeed a conditional probability distribution, i.e.\ $\sum_{s' \in \cS} \Pi(s'|s) = 1$ for all $s \in \cS$. This is equivalent to the condition that $\Pi = U \diag(1, (1+\lambda_2)/(1-\lambda_2), \ldots) U^{\dagger}$ is stochastic, or equivalently that $\Pi (1, 1, \ldots, 1)^{\dagger} = (1, 1, \ldots, 1)^{\dagger}$. 
  However, since $M$ is stochastic, we conclude that $(1, 1, \ldots, 1)^{\dagger}$ must be an eigenvector of $M$ with eigenvalue $1$; thus, in particular $U^{\dagger} (1, 1, \ldots, 1)^{\dagger} = (|\cS|, 0, 0, \ldots, 0)^{\dagger}$ since the unity eigenvalue is unique for $M$. This implies that $(1, 1, \ldots, 1)^{\dagger}$ is indeed an eigenvector of $\Pi$ with eigenvalue $1$.
  
  Next, let us verify~\eqref{eq:vdoublestar}. First, note that
  \begin{align*}
    \Cov_{S_1 \leftarrow \pi} [C_{S_1}, C_{S_{1+k}}] &= \sum_{s,s' \in \cS} \pi(s) \big[M^k\big]_{s,s'} C_{s} C_{s'} - C(\pi)^2 \\
    &= \sum_{s,s' \in \cS} \pi(s) \left( \left[ U \diag(1, \lambda_2^k, \ldots, \lambda_{|\cS|}^k) U^{\dagger} \right]_{s,s'} - \pi(s') \right) C_s C_{s'} \\
    &= \sum_{s,s' \in \cS} \pi(s) \left[ U \diag(0, \lambda_2^k, \ldots, \lambda_{|\cS|}^k) U^{\dagger} \right]_{s,s'} C_s C_{s'} ,
  \end{align*}
  where, in the final step, we used that $\pi(s') = \lim_{k \to \infty} [ M^k ]_{s,s'} = \big[ U \diag(1,0,\ldots,0) U^{\dagger} \big]_{s,s'}$ for all $s' \in \cS$. We can now easily sum these terms over $k$ using the geometric series to find
  \begin{align}
    \sum_{k=1}^\infty \Cov_{S_1 \leftarrow \pi} \big[ C_{S_1}, C_{S_{1+k}} \big] = \sum_{s,s' \in \cS} \pi(s) \left[ U \diag\left(0, \frac{\lambda_2}{1-\lambda_2}, \ldots, \frac{\lambda_{|\cS|}}{1-\lambda_{|\cS|}}\right) U^{\dagger} \right]_{s,s'} C_s C_{s'} \label{eq:covterms}
  \end{align}
  We also note that we may write
  \begin{align}
    \Var_{S \leftarrow \pi} [ C_S ] = \sum_{s,s' \in \cS} \pi(s) \big[ U \diag(1, 1, \ldots, 1) U^{\dagger} ]_{s,s'}\, C_s C_{s'}
    \label{eq:varterms}
  \end{align}
  Adding the expressions in~\eqref{eq:covterms} and~\eqref{eq:varterms} yields the desired expression for $V^{**}(M)$.
\end{proof}

%
%

\subsection{Generalized Inverse}
\label{app:inverse}

For any monotonically non-decreasing function $f: \mathbb{R} \to \mathbb{R} \cup \{-\infty, \infty\}$, we define its generalized inverse $f^{-1}$ on $\mathbb{R}$ as
\begin{align}
  f^{-1}: \mathbb{R} \to \mathbb{R} \cup \{-\infty, \infty\}, \quad y \mapsto \sup \{ x \in \mathbb{R} \,|\, f(x) \leq y \} . \label{eq:inv}
\end{align}
Note that this definition does not require the function to be continuous; however, if the function is upper semi-continuous (everywhere), the inverse satisfies the following useful properties.

\begin{lemma}
\label{lm:inverse}
Let $f, g: \mathbb{R} \to \mathbb{R} \cup \{-\infty, \infty\}$ be monotonically non-decreasing and upper semi-continuous. Then, the following holds:
\begin{enumerate}
\item[(a)] $f^{-1}$ is monotonically non-decreasing and upper semi-continuous. 
\item[(b)] We have $(f^{-1})^{-1} \equiv f$ wherever $f$ is finite.
\item[(c)] For any interval $\mathcal{I} \subseteq \mathbb{R}$, we have that $f(x) \geq g(x)$ for all $x\in \mathcal{I}$ implies  $f^{-1}(y) \leq g^{-1}(y)$ for all $y\in f(\mathcal{I})$. 
\end{enumerate}
\end{lemma}

\begin{proof}
  To show (a), we simply note that monotonicity and upper semi-continuity of $f^{-1}$ directly follow from the definition in~\eqref{eq:inv}.
  
  To show (b), we take any $x$ such that $f(x)$ is finite and write
  \begin{align*}
    (f^{-1})^{-1}(x) = \sup \big\{ y' \in \mathbb{R} \,\big|\, \sup \{ x' \in \mathbb{R} \,|\, f(x') \leq y' \} \leq x \}
  \end{align*}
  For any $\delta > 0$, we have the following. First, since 
  $$\sup \{ x' \in \mathbb{R} \,|\, f(x') \leq f(x) - \delta \} = \inf \{ x' \in \mathbb{R} \,|\, f(x') > f(x) - \delta \} \leq x,$$ we find that $y' = f(x) - \delta$ is feasible and thus $(f^{-1})^{-1}(x) \geq f(x) - \delta$. Second, upper semi-continuity implies that there exists a $\mu > 0$ such that $f(x + \mu) \leq f(x) + \delta$. Thus,
  $$\sup \{ x' \in \mathbb{R} \,|\, f(x') \leq f(x) + \delta \} \geq x + \mu > x. $$ 
  Hence, we have $(f^{-1})^{-1}(x) \leq f(x) + \delta$. The result then follows as $\delta > 0$ is arbitrary. 
  
  To show (c), we may write $f^{-1}(y) = \sup \{ x \in \mathcal{I} \,|\, f(x) \leq y \}$ since $y \in f(\mathcal{I})$. Thus,
  \begin{align*}
    f^{-1}(y) = \sup \{ x \in \mathcal{I} \,|\, f(x) \leq y \} \leq \sup \{ x \in \mathcal{I} \,|\, g(x) \leq y \} \leq 
    \sup \{ x \in \mathbb{R} \,|\, g(x) \leq y \} = g^{-1}(y) 
  \end{align*}
  which proves (c).
\end{proof}

\subsection{Basic Concentration Bounds}
\label{app:concentations}

\begin{lemma} \label{lem:conc1}
Let $S^n=(S_1, \ldots, S_n)$ be an i.i.d.\ random vector generated from $\pi\in\cP(\cS)$. Let $T_{S^n}$ be its type. For every $\eta>0$, 
\begin{align*}
\Pr\Big[ \|T_{S^n}-\pi \|_{\infty}>\eta\Big]\le 2|\cS| \exp(-2 n\eta^2 )
\end{align*}
\end{lemma}
\begin{proof}
This follows straightforwardly by the union bound and Hoeffding's inequality.
\end{proof}

Let $S^n= (S_1,\ldots, S_n)$ be a Markov chain on a finite state space $\cS$ that satisfies the {\em Doeblin condition}. That is, there exists an integer $m\ge 1$, a $\kappa>0$ and a probability measure $\mathbb{Q}$ on $\cS$ such that
\begin{equation}
\Pr[S_m \in A \, | \,  S_0 = s ]\ge \kappa\, \mathbb{Q} (A), \qquad \forall\, A\subset \cS , s\in\cS. \label{eqn:doeb}
\end{equation}
From \cite{Kont05}, we have the following generalization of Hoeffding's inequality for ergodic Markov chains:
\begin{theorem} \label{thm:kont}
Suppose the Markov chain   $S^n= (S_1,\ldots, S_n)$   has   stationary distribution $\pi  \in\cP(\cS)$ and satisfies the Doeblin condition in \eqref{eqn:doeb}. For any bounded function $F:\cS\to\mathbb{R}$ and any $\epsilon>0$, we have 
\begin{equation}
\Pr \Bigg[ \frac{1}{n}\sum_{k=1}^n \Big[ F(S_k) - \Exp_{S\leftarrow\pi}[F(S)]  \Big]\ge\epsilon\Bigg]\le\exp\left( -\frac{n-1}{2}\left[\frac{\kappa\epsilon}{m\bar{F}}-\frac{3}{n-1}\right]^2\right) \label{eqn:conc_mc}
\end{equation}
provided $n\ge 1 + 3m\bar{F}/(\kappa\epsilon)$ and where $\bar{F}:=\sup_s |F(s)|$.
\end{theorem}
Observe that we can take the measure $\mathbb{Q}$ to be the stationary distribution $\pi$ and the constant  $\kappa$ arbitrarily close to $1$, say $\kappa=\frac{1}{2}$. Using this theorem, we can provide an analogue of Lemma~\ref{lem:conc1} for well-behaved Markov chains.
\begin{lemma} \label{lem:markov_conc}
Let $S^n=(S_1, \ldots, S_n)$ be a Markov random vector satisfying the conditions in Theorem~\ref{thm:kont}. Let $T_{S^n}$ be its type. For every $\eta>0$, 
\begin{align*}
\Pr\Big[ \|T_{S^n}-\pi \|_{\infty}>\eta\Big]\le 2|\cS| \exp\left(-(n-1) \frac{\eta^2}{32m^2}   \right)
\end{align*}
for all $n\ge 12m/\eta+1$ where $m$ is the constant   in the Doeblin condition corresponding to $\kappa=\frac{1}{2}$.
\end{lemma}
\begin{proof}
This follows straightforwardly by the union bound and taking $F$ in Theorem~\ref{thm:kont}  to be various indicator functions such as $F_{s'}(s) := 1\{ s=s' \}$ for some $s'\in\cS$. We omit the details.
\end{proof}

\subsection{Uniform Bounds  on the Derivatives of the Gaussian PDF} \label{app:gaussian_der}
\begin{lemma}\label{lem:gaussian_der}
Let $\phi^{(k)}(x)$ be the $k$-th derivative of the Gaussian PDF $\phi(x)=\frac{1}{\sqrt{2\pi}}\exp\big(-x^2/2\big)$. Then 
\begin{align*}
\sup_{x\in\mathbb{R}} \big|\phi^{(k)}(x) \big| \le \frac{e^{1/8}}{ ( 2\pi)^{1/4}}  \, k^{1/4}\,  \left(\frac{k}{e}\right)^{k/2}
\end{align*}
for all $k\in\mathbb{N}$.
\end{lemma}
Before we prove Lemma~\ref{lem:gaussian_der}, let us  define the  {\em probabilists' Hermite polynomials}  and the {\em physicists' Hermite polynomials} as
\begin{align}
He_k(x) := (-1)^k e^{x^2/2}\frac{\mathrm{d}^k}{\mathrm{d} x^k}e^{-x^2/2} ,\quad\mbox{and}\quad  H_k(x):=(-1)^k e^{x^2 }\frac{\mathrm{d}^k}{\mathrm{d} x^k}e^{-x^2 }   \nonumber
\end{align}
respectively. Statement 22.5.18  in \cite{Abr65} asserts  $H_k(x)$ and $He_k(x)$ are related as follows:
\begin{equation}
He_k(x) = 2^{-k/2}H_k\bigg(\frac{x}{\sqrt{2}}\bigg). \label{eqn:relation}
\end{equation}
Moreover, Statement 22.14.17  in \cite{Abr65} also asserts that  the following bound on the physicists' Hermite polynomials holds for all $k\in \mathbb{N}$ and all $x\in\mathbb{R}$
\begin{equation}
 |H_k(x)| \le e^{1/12} \,  e^{x^2/2} \,  2^{k/2} \,  \sqrt{k!}.  \label{eqn:abr_bound}
\end{equation}

\begin{proof}[Proof of Lemma \ref{lem:gaussian_der}]
Now, by the definition of the  probabilists' Hermite polynomials, we see that the following holds 
\begin{equation}
 \phi^{( k )}(x) =   (-1)^k He_k(x)\phi(x).\nonumber
\end{equation}
We would like to bound $\sup_x |\phi^{( k )}(x) |$.  We have 
\begin{align}
\sup_x |\phi^{( k )}(x) | &= \frac{1}{\sqrt{2\pi}}\sup_x |He_k(x)|e^{-x^2/2}  =\frac{2^{-k/2} }{\sqrt{2\pi}}\sup_x  \bigg| H_k\bigg(\frac{x}{\sqrt{2}}\bigg) \bigg|   e^{-x^2/2}   \nonumber\\ & \le   e^{1/12} \, \frac{2^{-k/2} }{\sqrt{2\pi}}\sup_x  e^{x^2/4} \,  2^{k/2} \, \sqrt{k!} \, e^{-x^2/2} = \frac{e^{1/12} }{\sqrt{2\pi}} \sqrt{k!}    \nonumber
\end{align}
where in the second equality we used the relation in \eqref{eqn:relation} and for the inequality we used the bound in \eqref{eqn:abr_bound}.  Now, we further appeal to Stirling's upper bound on the factorial \cite[Statement 6.1.38]{Abr65} to obtain 
\begin{align}
\sup_x |\phi^{( k )}(x) | \le \frac{e^{1/12}}{\sqrt{2\pi}} \left(e^{1/12} \,  \sqrt{2\pi k } \left(\frac{k}{e}\right)^k  \right)^{1/2}  =\frac{e^{1/8}}{ ( 2\pi)^{1/4}} \,  k^{1/4} \, \left(\frac{k}{e}\right)^{k/2}  , \nonumber
\end{align}
as desired. 
\end{proof}

\subsection*{Acknowledgements}
MT is funded by the Ministry of Education (MOE) and National Research Foundation Singapore, as well as MOE Tier 3 Grant "Random numbers from quantum processes" (MOE2012-T3-1-009). VT would like to
acknowledge funding support  from NUS startup grants WBS R-263-000-A98-750 (FoE) and WBS  R-263-000-A98-133 (ODPRT) as well as  A*STAR, Singapore.

\end{document}